\documentclass[journal,10pt,twocolumn]{IEEEtran}

\usepackage{cite}

\ifCLASSINFOpdf
  \usepackage[pdftex]{graphicx}
  \graphicspath{{figures/}}
  \DeclareGraphicsExtensions{.pdf}
\else
  \usepackage[dvips]{graphicx}
  \graphicspath{{figures/}}
  \DeclareGraphicsExtensions{.eps}
\fi

\usepackage[cmex10]{amsmath}
\usepackage{amssymb}
\usepackage{amsthm}

\newtheorem{theorem}{Theorem}
\newtheorem{lemma}[theorem]{Lemma}

\usepackage{algorithmic}
\usepackage{algorithm}

\usepackage{verbatim}

\usepackage{url}
\usepackage{rotating}

\usepackage{color}

\newcommand {\MEL} {{\mathrm{M}}}
\newcommand {\R} {{\mathbb{R}}}
\newcommand {\Z} {{\mathbb{Z}}}
\newcommand {\E} {{\mathbb{E}}}
\newcommand {\om} {{\omega}}
\newcommand {\la} {{\lambda}}
\newcommand {\La} {{\Lambda}}

\newcommand {\lau} {{\lambda_1}}
\newcommand {\lad} {{\lambda_2}}

\newcommand {\quef} {{q}}
\newcommand {\quefu} {{q_1}}

\newcommand {\lat} {{\lambda_3}}

\newcommand {\bV} {{\bf V}}

\newcommand {\mm} {{l}}
\newcommand {\tS} {{\widetilde{S}}}
\newcommand {\Ufreq} {{U^{\mathrm{fr}}}}
\newcommand {\Sfreq} {{S^{\mathrm{fr}}}}
\newcommand {\freq} {{{\mathrm{fr}}}}

\begin{document}

\title{Deep Scattering Spectrum}

\author{Joakim~And\'en,~\IEEEmembership{Student Member,~IEEE,}
        St\'ephane~Mallat,~\IEEEmembership{Fellow,~IEEE}
\thanks{This work is supported by the ANR 10-BLAN-0126 and ERC InvariantClass 320959 grants.}
}

\markboth{Transactions on Signal Processing}%
{And\'en \& Mallat: Deep Scattering Spectrum}

\maketitle

\begin{abstract}
A scattering transform defines a locally translation invariant representation
which is stable to time-warping deformations. It extends MFCC representations
by computing modulation spectrum coefficients of multiple orders, through
cascades of wavelet convolutions and modulus operators. Second-order
scattering coefficients characterize transient phenomena such as attacks and
amplitude modulation. A frequency transposition invariant representation is
obtained by applying a scattering transform along log-frequency.
State-the-of-art classification results are obtained for musical genre and
phone classification on GTZAN and TIMIT databases, respectively.
\end{abstract}

\begin{IEEEkeywords}
Audio classification, deep neural networks, MFCC, modulation spectrum, wavelets.
\end{IEEEkeywords}

\IEEEpeerreviewmaketitle

\section{Introduction}

A major difficulty of audio representations for classification is the
multiplicity of information at different time scales: pitch and timbre at the
scale of milliseconds, the rhythm of speech and music at the scale of seconds,
and the music progression over minutes and hours. Mel-frequency cepstral
coefficients (MFCCs) are efficient local descriptors at time scales up to
$25~\mathrm{ms}$. Capturing larger structures up to $500~\mathrm{ms}$ is however
necessary in most applications. This paper studies the construction of stable,
invariant signal representations over such larger time scales. We concentrate
on audio applications, but introduce a generic scattering representation for
classification, which applies to many signal modalities beyond audio
\cite{Abry}.

Spectrograms compute locally time-shift invariant descriptors over durations limited
by a window. However, Section \ref{sec:mfsc} shows that high-frequency
spectrogram coefficients are not stable to variability due to time-warping
deformations, which occur in most signals, particularly in audio.
Stability means that small signal deformations produce small modifications of
the representation, measured with a Euclidean norm. This is particularly
important for classification.
Mel-frequency spectrograms are obtained by averaging spectrogram values over
mel-frequency bands. It improves stability to time warping, but it also removes
information. Over time intervals larger than $25~\mathrm{ms}$, the information
loss becomes too important, which is why mel-frequency spectrograms and MFCCs,
are limited to such short time
intervals. Modulation spectrum decompositions
\cite{hermansky,vinton2001scalable,mcdermott,ramona:2011,slaney,patterson,lee-shih,ellis-mcdermott,thompson2003non}
characterize the temporal evolution of mel-frequency spectrograms over larger
time scales, with autocorrelation or Fourier coefficients. However, this
modulation spectrum also suffers from instability to time-warping deformation,
which impedes classification performance.

Section \ref{sec:scattering} shows that the information lost by mel-frequency
spectrograms can be recovered with multiple layers of wavelet coefficients. In
addition to being locally invariant to time-shifts, this representation is
also stable to time-warping deformation. Known as a scattering transform
\cite{stephane}, it is computed through a cascade of wavelet transforms and
modulus non-linearities. The computational structure is similar to a
convolutional deep neural network
\cite{lecun,lee,hinton2012deep,deng2013deep,graves2013speech,humphrey2012learning,hamel2010learning,battenberg2012analyzing},
but involves no learning. It outputs time-averaged coefficients, providing
informative signal invariants over potentially large time scales.

A scattering transform has striking similarities with physiological models of
the cochlea and of the auditory pathway \cite{dau,shamma}, also used for audio
processing \cite{mesgarani:2006}. Its
energy conservation and other mathematical properties are reviewed in Section
\ref{sec:scattprop}. An approximate inverse scattering transform is introduced
in Section \ref{sec:inverse}, with numerical examples. Section \ref{sec:model}
relates the amplitude of scattering coefficients to audio signal properties.
These coefficients provide accurate measurements of frequency intervals
between harmonics and also characterize the amplitude modulation of
voiced and unvoiced sounds. The logarithm of scattering coefficients linearly
separates audio components related to pitch, formant and timbre.

\begin{figure*}[t!]
\centering
\setlength{\unitlength}{1in}
\begin{picture}(7.0,0.95)
\put(0.0,0.13){\includegraphics[width=7in]{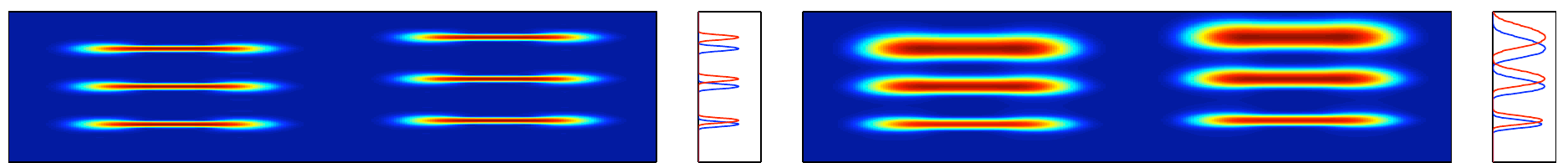}}
\put(0.0,0.86){\footnotesize$\omega$}
\put(3.55,0.86){\footnotesize$\omega$}
\put(2.98,0.1){\footnotesize$t$}
\put(6.52,0.1){\footnotesize$t$}
\put(1.45,0.0){\footnotesize{}(a)}
\put(5,0.0){\footnotesize{}(b)}
\color{blue}
\put(0.80,0.03){\line(0,1){0.85}}
\put(0.85,0.01){\footnotesize$t_0$}
\color{red}
\put(2.25,0.03){\line(0,1){0.85}}
\put(2.30,0.01){\footnotesize$t_1$}
\color{blue}
\put(4.33,0.03){\line(0,1){0.85}}
\put(4.38,0.01){\footnotesize$t_0$}
\color{red}
\put(5.78,0.03){\line(0,1){0.85}}
\put(5.83,0.01){\footnotesize$t_1$}
\end{picture}
\caption{\label{fig:dilation-spec}{ (a) Spectrogram $\log |\widehat{x}(t,\omega)|$ for a harmonic signal $x(t)$ (centered in $t_0$) followed by
$\log |\widehat{x}_\tau (t,\omega)|$ for $x_\tau(t) = x((1-\epsilon)t)$ (centered
in $t_1$), as a function of $t$ and $\omega$. The right graph plots $\log|\widehat{x}(t_0, \omega)|$ (blue) and $\log |\widehat{x}_\tau (t_1, \omega)|$ (red) as a
function of $\omega$. Their partials do not overlap at high frequencies. (b)
Mel-frequency spectrogram $\log M x(t,\omega)$ followed by $\log M x_\tau
(t,\omega)$. The right graph plots $\log M x(t_0, \omega)$ (blue) and $\log M
x_\tau (t_1, \omega)$ (red) as a function of $\omega$. With a mel-scale
frequency averaging, the partials of $x$ and $x_\tau$ overlap at all
frequencies. }}
\end{figure*}
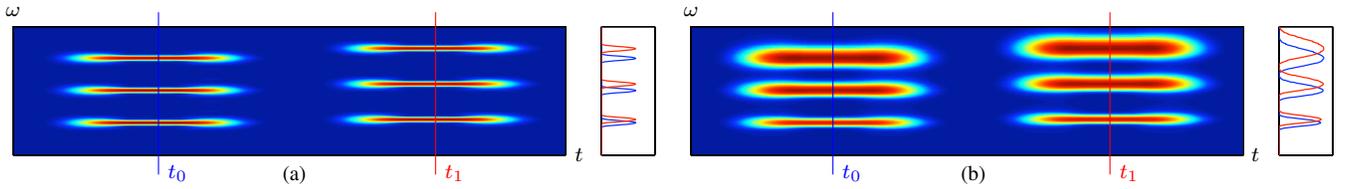

Frequency transpositions form another important source of audio variability,
which should be kept or removed depending upon the classification task.
For example, speaker-independent phone classification requires some frequency
transposition invariance, while frequency localization is necessary for speaker identification. Section \ref{sec:cls} shows that cascading a
scattering transform along log-frequency yields a transposition invariant
representation which is stable to frequency deformation.

Scattering representations have proved useful for image classification
\cite{joan:2012,laurent:2013}, where spatial translation invariance is
crucial. In audio, the analogous time-shift invariance is also important, but
scattering transforms are computed with very different wavelets. They have a
better frequency resolution, which is adapted to audio frequency structures.
Section \ref{sec:classif} explains how to adapt and optimize the
frequency invariance for each signal class at the supervised learning stage.
A time and frequency scattering representation is used for musical genre
classification over the GTZAN database, and for phone segment classification
over the TIMIT corpus. State-of-the-art results are obtained with a Gaussian
kernel SVM applied to scattering feature vectors. All figures and results are
reproducible using a MATLAB software package, available at {\it
http://www.di.ens.fr/data/scattering/}.

\section{Mel-frequency Spectrum}
\label{sec:mfsc}

Section \ref{Fourernsinfs} shows that high-frequency spectrogram coefficients
are not stable to time-warping deformation. The mel-frequency spectrogram
stabilizes these coefficients by averaging them along frequency, but loses
information. To analyze this information loss, Section \ref{melsdfns} relates
the mel-frequency spectrogram to the amplitude output of a filter bank which
computes a wavelet transform.

\subsection{Fourier Invariance and Deformation Instability}
\label{Fourernsinfs}

Let $\widehat{x}(\om) = \int x(u)e^{-i\om{}u}du$ be the Fourier transform of $x$.
If $x_c(t) = x(t-c)$ then $\widehat{x}_c(\om) = e^{-ic\om}\,\widehat{x}(\om)$. The
Fourier transform modulus is thus invariant to translation:
\begin{equation}
	|\widehat{x}_c(\omega)| = |\widehat{x}(\omega)|~.
\end{equation}
A spectrogram localizes this translation invariance with a window $\phi$ of
duration $T$ such that $\int \phi(u)du = 1$. It is defined by
\begin{equation}
\label{spectro}
	|\widehat{x}(t,\om)| = \left| \int x(u)\,\phi(u-t)\,e^{-i\om{}u}\,du\right|~.
\end{equation}
If $|c| \ll T$ then one can verify that $|\widehat{x}_c(t,\om)| \approx
|\widehat{x}(t,\om)|$.

However, invariance to time-shifts is often not enough. Suppose that $x$ is not just translated but time-warped to give $x_\tau(t) =
x(t-\tau(t))$ with $|\tau'(t)| < 1$. A representation $\Phi(x)$ is said to be
stable to deformation if its Euclidean norm $\|\Phi(x) - \Phi(x_\tau)\|$ is
small when the deformation is small. The deformation size is measured by
$\sup_t |\tau'(t)|$. If it vanishes then it is a ``pure'' translation without
deformation. Stability is formally defined as a Lipschitz continuity condition
relatively to this metric. It means that there exists $C > 0$ such that for
$x(t)$ and all $\tau$ with $\sup_t |\tau'(t)| < 1$
\begin{equation}
\label{eq:stability} 
\|\Phi(x) - \Phi(x_\tau)\| \leq C\, \sup_t |\tau'(t)|\, \|x\|~.
\end{equation}
The constant $C$ is a measure of stability.

This Lipschitz continuity property implies that time-warping deformations are
locally linearized by $\Phi(x)$. Indeed, Lipschitz continuous operators are
almost everywhere differentiable. It results that $\Phi(x) - \Phi(x_\tau)$ can
be approximated by a linear operator if $\sup_t |\tau'(t)|$ is small. A family
of small deformations thus generate a linear space. In the transformed space,
an invariant to these deformations can then be computed with a linear
projector on the orthogonal complement to this linear space. In Section
\ref{sec:classif} we use linear discriminant classifiers to become selectively
invariant to small time-warping deformations.

A Fourier modulus representation $\Phi(x) = |\widehat x|$ is not stable to
deformation because high frequencies are severely distorted by small
deformations. For example, let us consider a small dilation $\tau(t) =
\epsilon{}t$ with $0 < \epsilon \ll 1$. Since $\tau'(t) = \epsilon$, the
Lipschitz continuity condition
\eqref{eq:stability} becomes
\begin{equation}
\label{contconds}
	\||\widehat{x}|-|\widehat{x_\tau}|\| \leq C\,\epsilon\,\|x\|~.
\end{equation}
The Fourier transform of $x_\tau(t) = x((1-\epsilon)t)$ is
$\widehat{x}_\tau(\om) = (1-\epsilon)^{-1}\,\widehat{x}((1-\epsilon)^{-1}\om)$.
This dilation shifts a frequency component at $\om_0$ by $\epsilon|\om_0|$.
For a harmonic signal $x(t) = g(t)\sum_n
a_n\cos(n\xi{}t)$, the Fourier transform is a sum of partials
\begin{equation}
	\widehat{x}(\om) = \sum_n \frac{a_n}{2}\Big( \widehat{g}(\om-n\xi)+\widehat{g}(\om+n\xi)\Big).
\end{equation}
After time-warping, each partial $\widehat g(\omega \pm n\xi)$ is translated
by $\epsilon{}n \xi$, as shown in the spectrogram of Figure
\ref{fig:dilation-spec}(a). Even though $\epsilon$ is small, at high
frequencies $n \epsilon \xi$ becomes larger than the bandwidth of $\widehat g$.
Consequently, the harmonics $\widehat
g(\omega(1-\epsilon)^{-1}-n\xi)$ of $\widehat x_\tau$ do not overlap with the 
harmonics $\widehat g(\omega - n\xi)$ of $\widehat x$.
The Euclidean distance of $|\widehat x|$ and $|\widehat x_\tau|$
thus does not decrease proportionally to $\epsilon$ if the harmonic amplitudes
$a_n$ are sufficiently large at high frequencies. This proves that
the deformation stability 
condition \eqref{contconds} is not satisfied for any $C > 0$.

The autocorrelation $R x (u) = \int x(t)\, x^{\star}(t-u)\,dt$ is also a
translation invariant representation which has the same deformation
instability as the Fourier transform modulus. Indeed, $\widehat Rx (\omega) =
|\widehat x(\omega)|^2$ so $\|Rx - R x_\tau \| = (2 \pi)^{-1} \||\widehat x|^2
- |\widehat x_\tau|^2 \|$.

\subsection{Mel-frequency Deformation Stability and Filter Banks}
\label{melsdfns}

A mel-frequency spectrogram averages the spectrogram energy with mel-scale
filters $\widehat{\psi}_\la$, where $\la$ is the center frequency of each $\widehat
\psi_{\la} (\omega)$:
\begin{equation}
	\label{eq:mfsc-def}
	\MEL x(t,\la) = \frac{1}{2\pi} \int |\widehat{x}(t,\omega)|^2\,|\widehat{\psi}_{\la}(\om)|^2d\om~.
\end{equation}
The band-pass filters $\widehat \psi_\la$ have a constant-$Q$ frequency bandwidth at high frequencies. Their frequency
support is centered at $\la$ with a bandwidth of the order of $\la/Q$. 
At lower frequencies, instead of being constant-Q, the bandwidth of 
$\widehat \psi_\la$ remains equal to $2 \pi / T$.

The mel-frequency averaging removes deformation instability created by large
displacements of high frequencies under dilations. If $x_\tau (t) =
x((1-\epsilon)t)$ then we saw that each frequency component at $\omega_0$ is
moved by $\epsilon |\omega_0|$, which may be large if $|\omega_0|$ is large.
However, the mel-scale filter $\widehat{\psi}_\lambda (\om)$ covering the
frequency $\om_0$ has a frequency bandwidth of the order of $\lambda/Q \sim
|\om_0|/Q$. As a result, the relative error after averaging by
$|\widehat{\psi}|^2$ is of the order of $\epsilon{}Q$. This is illustrated by
Figure \ref{fig:dilation-spec}(b) on a harmonic signal $x$. After
mel-frequency averaging, the frequency partials of $x$ and $x_\tau$ overlap at
all frequencies. One can verify that $\|Mx(t,\lambda) - M_\tau x(t,\lambda)\|
\leq C\,\epsilon \|x\|$, where $C$ is proportional to $Q$, and does not depend
upon $\epsilon$ and $x$. Unlike the spectrogram (\ref{spectro}), the
mel-frequency spectrogram (\ref{eq:mfsc-def}) satisfies the Lipschitz
deformation stability condition \eqref{eq:stability}.

Mel-scale averaging provides time-warping stability but
loses information. We show that this frequency averaging is equivalent to
a time averaging of a filter bank output, which will provide a strategy to
recover the lost information. Since
$\widehat{x}(t,\om)$ in \eqref{spectro} is the Fourier transform of
$x_t(u) = x(u)\phi(u-t)$, applying Plancherel's formula gives
\begin{align}
	\MEL x(t,\la) &= \frac{1}{2\pi} \int |\widehat{x_t}(\om)|^2\,
|\widehat{\psi}_{\la}(\om)|^2\, d\om \\
	&= \int |x_t\star\psi_{\la}(v)|^2\,dv \\
	&= \int \left|\int x(u)\phi(u-t)\psi_\la(v-u)du\right|^2\,dv
\end{align}
If $\la \gg Q/T$ then $\phi(t)$ is approximately constant on the support of
$\psi_\la(t)$, so $\phi(u-t)\psi_\la(v-u) \approx \phi(v-t)\psi_\la(v-u)$,
and hence
\begin{align}
	\MEL x(t,\la) &\approx \int \left|\int x(u)\psi_\la(v-u)du\right|^2|\phi(v-t)|^2dv \\
\label{Melwaves}
	&= |x\star\psi_\la|^2\star|\phi|^2(t)~.
\end{align}
The frequency averaging of the spectrogram is thus nearly equal to the time
averaging of $|x \star \psi_\la|^2$. In this formulation, the window $\phi$
acts as a lowpass filter, ensuring that the representation is locally
invariant to time-shifts smaller than $T$. Section \ref{sec:wavelet} studies
the properties of the constant-Q filter bank $\{ \psi_\la \}_\la$, which
defines an analytic wavelet transform.

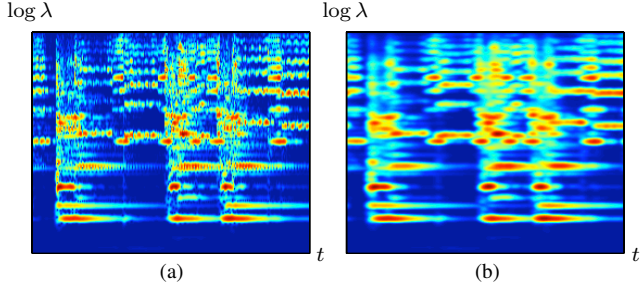
\begin{figure}[t]
\centering
\setlength{\unitlength}{1in}
\begin{picture}(3.3,1.6)
	\put(0.08,0.12){\includegraphics[width=3.2in]{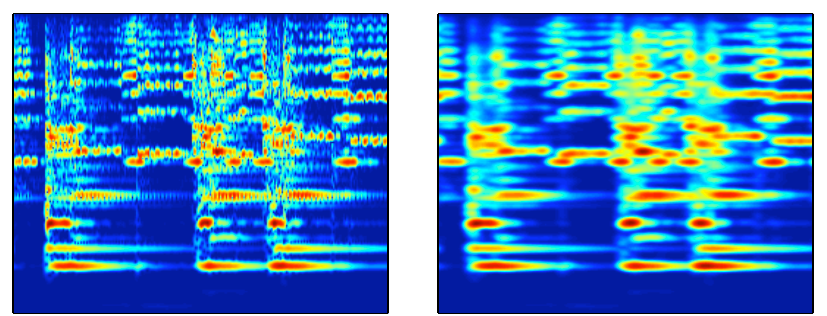}}
	\put(1.62,0.09){\footnotesize$t$}
	\put(3.27,0.09){\footnotesize$t$}
	\put(0.8,0){\footnotesize(a)}
	\put(2.45,0){\footnotesize(b)}
	\put(0.0,1.37){\footnotesize$\log \la$}
	\put(1.65,1.37){\footnotesize$\log \la$}
\end{picture}
\caption{\label{fig:scalogram-to-mfsc}{(a): Scalogram $\log |x \star \psi_\la (t)|^2$ for a musical signal, as a function of $t$ and $\la$. (b): Averaged scalogram $\log |x\star\psi_\la|^2\star\phi^2(t)$ with a
lowpass filter $\phi$ of duration $T = 190~\mathrm{ms}$.}}
\end{figure}

Figures \ref{fig:scalogram-to-mfsc}(a) and \ref{fig:scalogram-to-mfsc}(b)
display $|x\star\psi_\la|^2$ and $|x\star\psi_\la|^2 \star |\phi|^2$,
respectively, for a musical recording. The window duration is $T =
190~\mathrm{ms}$. This time averaging removes fine-scale information such as
vibratos and attacks. To reduce information loss, a mel-frequency spectrogram
is often computed over small time windows of about $25~\mathrm{ms}$. As a
result, it does not capture large-scale structures, which limits
classification performance.

To increase $T$ without losing too much information, it is necessary to
capture the amplitude modulations of $|x\star\psi_\la (t)|$ at scales smaller than $T$, which are important in audio perception. 
The spectrum of these modulation envelopes can
be computed from the spectrogram
\cite{hermansky,vinton2001scalable,mcdermott,ramona:2011} of $|x\star\psi_\la|$, or represented with a
short-time autocorrelation \cite{slaney,patterson}. However, these modulation
spectra are unstable to time-warping deformations. Indeed, a time-warping of
$x$ induces a time-warping of $|x\star\psi_\la|$, and Section
\ref{Fourernsinfs} showed that spectrograms and autocorrelations have
deformation instabilities. Constant-Q averaged modulation spectra
\cite{thompson2003non,ellis-mcdermott} stabilize spectrogram representations with
another averaging along modulation frequencies. According to \eqref{Melwaves},
this can also be computed with a second constant-Q filter bank. The scattering
transform follows this latter approach.

\section{Wavelet Scattering Transform}
\label{sec:scattering}
A scattering transform recovers the information lost by a mel-frequency
averaging with a cascade of wavelet decompositions and modulus operators
\cite{stephane}. It is locally translation invariant and stable to
time-warping deformation. Important properties of constant-Q filter banks are
first reviewed in the framework of a wavelet transform, and the scattering
transform is introduced in Section \ref{sec:scattnet}.

\subsection{Analytic Wavelet Transform and Modulus}
\label{sec:wavelet}
Constant-Q filter banks compute a wavelet transform. We review
the properties of complex analytic wavelet transforms and their modulus,
which are used to calculate mel-frequency spectral coefficients.

A wavelet $\psi (t)$ is a band-pass filter with $\widehat \psi(0) = 0$. We
consider complex wavelets with quadrature phase 
such that $\widehat \psi(\om) \approx 0$ for $\om < 0$.
For any $\la > 0$, a dilated wavelet of center frequency $\la$ is written
\begin{equation}
\label{psidilation}
\psi_\la (t) = \la\, \psi(\la\, t)~~\mbox{and hence}~~
\widehat \psi_\la (\om) = \widehat \psi\Bigl(\frac{\omega} {\la}\Bigr)~.
\end{equation}

The center frequency of $\widehat \psi$ is normalized to $1$. 
In the following, we denote
by $Q$ the number of wavelets per octave, which means that
$\lambda = 2^{k/Q}$ for $k \in \Z$.  The bandwidth of $\widehat \psi$
is of the order of $Q^{-1}$, 
to cover the whole frequency axis with these band-pass wavelet filters.
The support of $\widehat \psi_\la(\om)$ is centered in $\la$ with a frequency
bandwidth $\la/Q$ whereas the energy of $\psi_\la (t)$ is concentrated around
$0$ in an interval of size $2 \pi Q/\la$. To guarantee that this interval is
smaller than $T$, we define $\psi_\la$ with \eqref{psidilation} only for $\la
\geq 2 \pi Q/T$. For $\la < 2 \pi Q/T$, the lower-frequency interval $[0,2 \pi
Q/T]$ is covered with about $Q-1$ equally-spaced filters $\widehat \psi_\la$
with constant frequency bandwidth $2 \pi/T$. For simplicity, these
lower-frequency filters are still called wavelets. 
We denote by $\Lambda$ the grid of all wavelet center frequencies $\la$.

The wavelet transform of $x$ computes a convolution of $x$
with a low-pass filter
$\phi$ of frequency bandwidth $2 \pi/T$, and convolutions
with all higher-frequency wavelets $\psi_\la$ for $\lambda \in \Lambda$:
\begin{equation}
\label{wavensdfons}
	W x = \Bigl(x\star\phi(t)\,,\, x\star\psi_\la (t)
\Bigr)_{t \in \R,\la\in{\Lambda}}~.
\end{equation}
This time index $t$ is not critically sampled as in wavelet bases so this
representation is highly redundant. 
The wavelet $\psi$ and the low-pass filter
$\phi$ are designed to build filters which cover the whole frequency axis,
which means that
\begin{equation}
	\label{eq:paley-littlewood}
A(\omega) = 
|\widehat{\phi}(\omega)|^2+\frac{1}{2}\sum_{\la\in{\Lambda}} 
\Big(|\widehat{\psi}_\la(\omega)|^2+|\widehat{\psi}_\la(-\omega)|^2 \Big) ~
\end{equation}
satisfies, for all $\om \in \R$:
\begin{equation}
	\label{eq:paley-littlewood2}
1-\alpha \leq A(\omega) \leq 1~~\mbox{with}~~\alpha < 1~.
\end{equation}
This condition implies that the wavelet transform $W$ is
a stable and invertible operator. Multiplying (\ref{eq:paley-littlewood2}) by
$|\widehat{x}(\omega)|^2$ and applying the Plancherel formula
\cite{mallatbook} gives
\begin{equation}
	\label{eq:wavelet-norm}
	(1-\alpha)\|x\|^2 \le 
\|Wx\|^2 \le \|x\|^2~,
\end{equation}
where $\|x\|^2 = \int |x(t)|^2 dt$ and
where the squared norm of $W x$ sums all squared coefficients:
\[
\|Wx\|^2 = \int |x\star\phi(t)|^2 \, dt
+\sum_{\la\in{\Lambda}} \int |x\star\psi_\la (t)|^2\, dt ~.
\]
The upper bound (\ref{eq:wavelet-norm}) means that $W$ is a contractive
operator and the lower bound implies that it has a stable inverse. One can
also verify that the pseudo-inverse of $W$ recovers $x$ with the following
formula
\begin{equation}
	\label{eq:reproducing-kernel}
	x(t) = (x\star\phi)\star\overline{\phi}(t)+\sum_{\la\in{\Lambda}}
{\rm Real}\Big((x\star\psi_{\la})\star\overline{\psi}_{\la}(t)\Big)~,
\end{equation}
with reconstruction filters defined by
\begin{equation}
	\label{eq:filtern}
\widehat{\overline{\phi}}(\omega) = \frac{{\widehat{\phi}^*(\omega)}} {A(\omega)}
~~\mbox{and}~~
\widehat{\overline{\psi}}_\la(\omega) = \frac{{\widehat{\psi}^*_\la(\omega)}} {A(\omega)}~,
\end{equation}
where $z^*$ is the complex conjugate of $z \in \mathbb{C}$. 
If $\alpha = 0$ in (\ref{eq:paley-littlewood2}) 
then $W$ is said to be a tight frame operator, in which case
$\overline{\phi}(t) = \phi(-t)$ and $\overline{\psi}_\la(t) = \psi_\la^*(-t)$.

One may define an analytic wavelet with an octave resolution $Q$ as $\psi(t) =
e^{i t}\, \theta(t)$ and hence $\widehat \psi(\omega) = \widehat
\theta(\omega-1)$ where $\widehat \theta$ is the transfer function of a
low-pass filter whose bandwidth is of the order of $Q^{-1}$. If
$\widehat \theta(-1) \neq 0$ then we define
$\widehat \psi(\omega) = \widehat \theta(\omega-1) - \widehat
\theta(\omega) \widehat \theta(-1)/\widehat \theta(0)$, which guarantees that
$\widehat \psi(0) = 0$.
If $\theta$ is a Gaussian then $\psi$ is called a Morlet wavelet,
which is almost analytic 
because $|\widehat \psi(\om)|$ is small but not strictly zero for $\om < 0$.
Figure \ref{fig:gabor} shows Morlet wavelets $\widehat \psi_\la$ with $Q = 8$.
In this case $\phi$ is also chosen to be a Gaussian. 
For $Q = 1$, tight frame
wavelet transforms can also be obtained by choosing $\psi$ to be the analytic
part of a real wavelet which generates an orthogonal wavelet basis, such as a
cubic spline wavelet \cite{stephane}. Unless indicated otherwise, wavelets used in this paper are Morlet wavelets.

Following \eqref{Melwaves}, mel-frequency spectrograms can be approximated
using a non-linear wavelet modulus operator which removes the complex phase of
all wavelet coefficients:
\begin{equation}
\label{wtildedef}
	|W| x = \Bigl(x\star\phi(t)\,,\, |x\star\psi_\la (t)|
\Bigr)_{t \in \R,\la\in{\Lambda}}~.
\end{equation} 
Taking the modulus of analytic wavelet coefficient can be interpreted as a
sub-band Hilbert envelope demodulation. Demodulation is used to separate
carriers and modulation envelopes. When a carrier or pitch frequency can be
detected, then a linear coherent demodulation is efficiently implemented by
multiplying the analytic signal with the conjugate of the detected carrier
\cite{schimmel2005coherent,turner2011probabilistic,sell2010solving}. However,
many signals such as unvoiced speech are not modulated by isolated carrier
frequency, in which case coherent demodulation is not well defined. Non-linear
Hilbert envelope demodulations apply to any band-pass analytic signals, but if
a carrier is present then the Hilbert envelope depends both on the carrier and
on the amplitude modulation. Section \ref{ModusSpec} explains how to isolate
amplitude modulation coefficients from Hilbert envelope measurements, whether
a carrier is present or not.

Although a wavelet modulus operator removes the complex phase, it does not
lose information because the temporal variation of the multiscale envelopes is
kept. A signal cannot be reconstructed from the modulus of its Fourier
transform, but it can be recovered from the modulus of its wavelet transform.
Since the time variable $t$ is not subsampled, a wavelet transform has more
coefficients than the original signal. These coefficients are highly redundant
when filters have a significant frequency overlap. For particular families of
analytic wavelets, one can prove that $|W|$ is an invertible operator with a
continuous inverse \cite{Irene}. This is further studied in Section
\ref{sec:inverse}.

The operator $|W|$ is contractive. Indeed, the wavelet transform $W$ is
contractive and the complex modulus is contractive in the sense that $||a| -
|b| | \leq |a - b|$ for any $(a,b) \in \mathbb{C}^2$ so
\[
\|\,|W| x - |W| x' \|^2 \leq \|W x - W x' \|^2 \leq \|x - x' \|^2~.
\]
If $W$ is a tight frame operator then $\|\,|W| x\| = \|W x\| = \|x\|$ so $|W|$
preserves the signal norm.

\begin{figure}[t]
\centering
\setlength{\unitlength}{1in}
\begin{picture}(3.4,0.7)
\put(0,0.1){\includegraphics[width=3.3in]{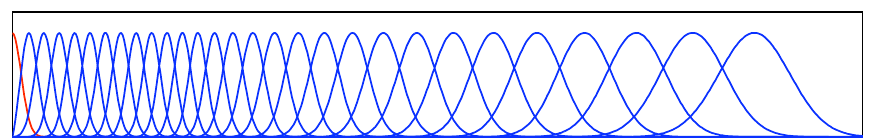}}
\put(3.29,0.08){\footnotesize $\omega$}
\end{picture}
\caption{\label{fig:gabor}{Morlet wavelets $\widehat \psi_{\la} (\om)$
with $Q = 8$ wavelets per octave, for different $\la$. The low-frequency
filter $\widehat \phi(\om)$ (in red) is a Gaussian.}}
\end{figure}

\subsection{Deep Scattering Network}
\label{sec:scattnet}

We showed in (\ref{Melwaves}) that mel-frequency spectral coefficients $M
x(t,\la)$ are approximately equal to averaged squared wavelet coefficients $|x
\star \psi_\la|^2 \star |\phi|^2(t)$. Large wavelet coefficients are
considerably amplified by the square operator. To avoid amplifying outliers,
we remove the square and calculate $|x\star\psi_{\la}|\star\phi(t)$ instead.
High frequencies removed by the low-pass filter $\phi$ are recovered by a new
set of wavelet modulus coefficients. Cascading this procedure defines a
scattering transform.

A locally translation invariant descriptors of $x$ is obtained with a
time-average $S_0 x(t) = x \star \phi(t)$, which removes all high frequencies.
These high-frequencies are recovered by a wavelet modulus transform
\[
|W_1| x = \Bigl(x\star\phi(t) \,,\, |x\star\psi_{\la_1}(t)|
\Bigr)_{t \in \R,\la_1\in\Lambda_1}~.
\]
It is computed with wavelets $\psi_\lau$ having an octave frequency resolution
$Q_1$. For audio signals we set $Q_1 = 8$, which defines wavelets having the
same frequency resolution as mel-frequency filters. Audio signals have little
energy at low frequencies so $S_0 x(t) \approx 0$. Approximate mel-frequency
spectral coefficients are obtained by averaging the wavelet modulus
coefficients with $\phi$:
\begin{equation}
	\label{eq:first-order-def}
S_1 x(t,\lau) = |x\star\psi_{\la_1}|\star\phi(t) ~.
\end{equation}
These are called first-order scattering coefficients. They are computed with a
second wavelet modulus transform $|W_2|$ 
applied to each $|x\star\psi_{\la_1}|$, which
also provides complementary high-frequency wavelet coefficients:
\[
|W_2|\, |x\star\psi_{\la_1}| = 
\Big(|x\star\psi_{\la_1}|\star\phi
\,,\, ||x\star\psi_{\la_1}| \star \psi_\lad |\Big)_{\la_2 \in \Lambda_2}
.
\]
The wavelets $\psi_\lad$ have an octave resolution $Q_2$ which may be
different from $Q_1$. It is chosen to get a sparse representation which means
concentrating the signal information over as few wavelet coefficients as
possible. These coefficients are averaged by the lowpass filter $\phi$ of size $T$, which ensures local invariance to time-shifts, as with the first-order coefficients. It defines second-order scattering coefficients:
\begin{equation*}
	\label{eq:second-order-def}
S_2 x(t,\la_1,\lad) = ||x\star\psi_{\la_1}| \star \psi_\lad| \star\phi(t) ~.
\end{equation*}
These averages are computed by applying a third wavelet modulus
transform $|W_3|$ to each $||x\star\psi_{\la_1}| \star \psi_\lad|$.
It computes their wavelet modulus coefficients through convolutions with 
a new set of wavelets $\psi_{\la_3}$ having an octave resolution $Q_3$.
Iterating this process defines scattering coefficients at any order $m$.

For any $m \geq 1$, iterated wavelet modulus convolutions are written:
\begin{equation}
\label{insdf8sdf8}
U_m x(t,\la_1,...,\la_m) = |\,||x \star \psi_\lau| \star... |\star \psi_{\la_m} (t)|~,
\end{equation}
where $m$th-order wavelets $\psi_{\la_m}$ have an octave resolution $Q_m$,
and satisfy the stability condition \eqref{eq:paley-littlewood2}. Averaging
$U_m x$ with $\phi$ gives scattering coefficients of order $m$:
\begin{eqnarray*}
S_m x (t,\la_1,...,\la_m) &=& 
|\,||x \star \psi_\lau| \star... |\star \psi_{\la_m} |\star \phi(t)\\
&=& U_m x(.,\la_1,...,\la_m) \star \phi(t)~.
\end{eqnarray*}
Applying $|W_{m+1}|$ on $U_m x$ computes both $S_m x$ and $U_{m+1} x$:
\begin{equation}
\label{indfsdf}
|W_{m+1}| \,U_{m} x = (S_{m} x \,,\,U_{m+1} x)~.
\end{equation}
A scattering decomposition of maximal order $\mm$ is thus defined by
initializing $U_0 x = x$, and recursively computing (\ref{indfsdf}) for $0
\leq m \leq \mm$. This scattering transform is illustrated in Figure
\ref{algo}. The final scattering vector aggregates all scattering coefficients
for $0 \leq m \leq \mm$:
\begin{equation}
S x = ( S_m x )_{0 \leq m \leq \mm}.
 \label{eq:scattvec}
\end{equation}

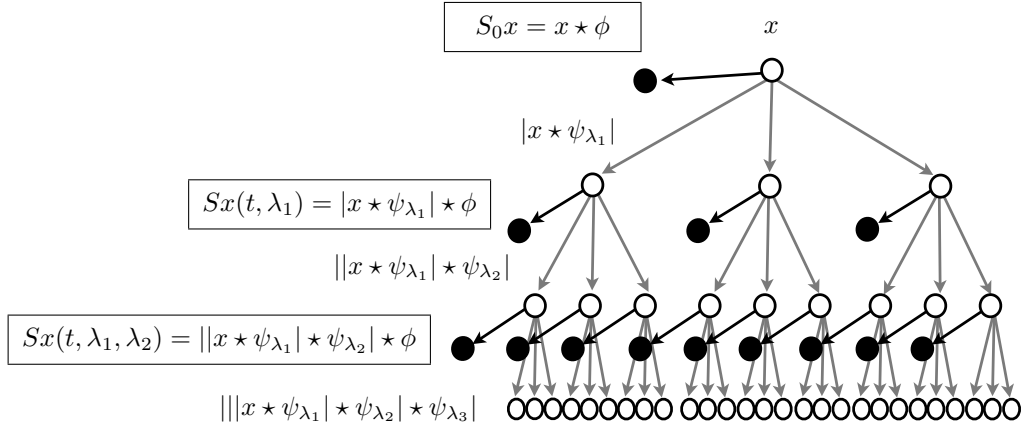
\begin{figure*}
\centering
\setlength{\unitlength}{0.4cm}
\begin{picture}(40,15)
	\put(15,0){\includegraphics[width=7.7cm]{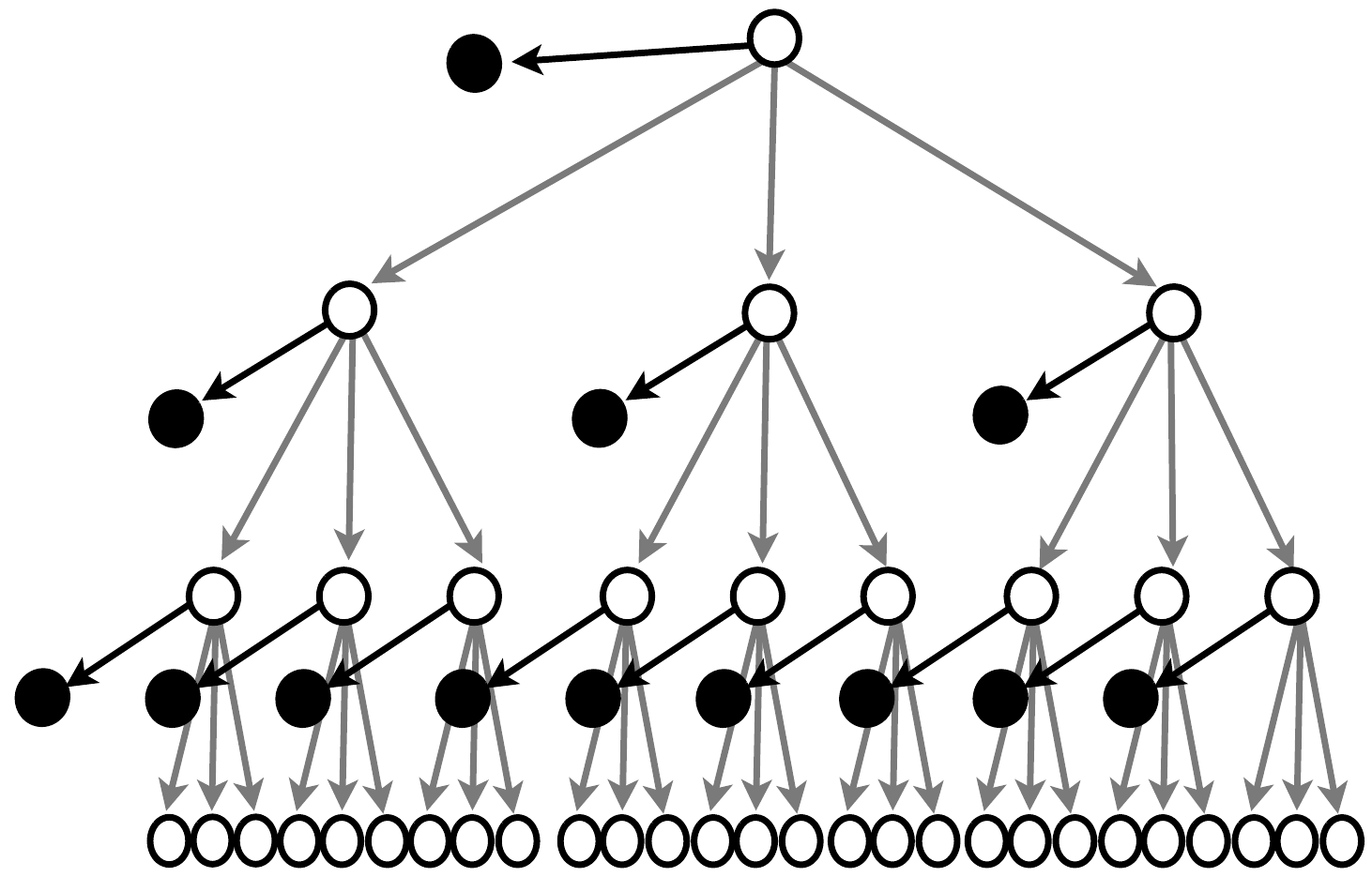}}
	\put(25.6,13){$x$}
	\put(15,12.4){\framebox(6.5,1.6){$S_0x = x\star\phi$}}
	\put(17.5,9.5){$|x\star\psi_\lau|$}
	\put(6.5,6.5){\framebox(10,1.6){$Sx(t,\lau) = |x\star\psi_\lau|\star\phi$}}
	\put(11.3,5){$||x\star\psi_\lau|\star\psi_\lad|$}
	\put(0.5,2){\framebox(14,1.6){$Sx(t,\lau,\lad) = ||x\star\psi_\lau|\star\psi_\lad|\star\phi$}}
	\put(7.5,0.3){$|||x\star\psi_\lau|\star\psi_\lad|\star\psi_\lat|$}
\end{picture}
\caption{A scattering transform iterates on wavelet modulus operators $|W_m|$
to compute cascades of $m$ wavelet convolutions and moduli stored in $U_m x$,
and to output averaged scattering coefficients $S_m x$.}
\label{algo}
\end{figure*}

The scattering cascade of convolutions and non-linearities can also be
interpreted as a convolutional network \cite{lecun}, where $U_m x$ is the set
of coefficients of the $m$th internal network layer. These networks have been
shown to be highly effective for audio classification
\cite{lee,hinton2012deep,deng2013deep,graves2013speech,humphrey2012learning,hamel2010learning,battenberg2012analyzing}.
However, unlike standard convolutional networks, each such layer has an output
$S_m x = U_m x \star \phi$, not just the last layer. In addition, all filters
are predefined wavelets and are not learned from training data. A scattering
transform, like MFCCs, provide a low-level invariant representation of the
signal, without learning. It relies on prior information concerning the type
of invariants that need to be computed, in this case relatively to time-shifts
and time-warping deformations, or in Section \ref{sec:cls} relatively to
frequency transpositions. When no such information is available, or if the
sources of variability are much more complex, then it is necessary to learn
them from examples, which is a task well suited for deep neural networks.
In that sense both approaches are complementary.

The wavelet octave resolutions are optimized at each layer $m$ to produce
sparse wavelet coefficients at the next layer. This better preserves the
signal information as explained in Section \ref{sec:inverse}. 
Sparsity seems also to play an important role 
for classification \cite{henaff2011unsupervised,nam2012learning}.
For audio signals $x$, 
choosing $Q_1 = 8$ wavelets per octave has been shown to provide
sparse representations of a mix of speech, music and environmental signals
\cite{smith-lewicki:2006}. It nearly corresponds to a mel-scale frequency
subdivision.

At the second order, choosing $Q_2 = 1$ defines wavelets with more narrow time
support, which are better adapted to characterize transients and
attacks. Section \ref{sec:model} shows
that musical signals including modulation structures such as tremolo may
however require wavelets having better frequency resolution, and hence $Q_2 >
1$. At higher orders $m \geq 3$ we always set $Q_m = 1$, but we shall see
that these coefficients can often be neglected.

The scattering cascade has similarities with several neurophysiological models
of auditory processing, which incorporate cascades of constant-Q filter banks
followed by non-linearities \cite{dau,shamma}. The first filter bank with $Q_1
= 8$ models the cochlear filtering, whereas the second filter bank corresponds
to later processing in the models with filters that have $Q_2 = 1$
\cite{dau,shamma}.

\section{Scattering Properties}
\label{sec:scattprop}

We briefly review important properties of scattering transforms, including
stability to time-warping deformation, energy conservation, and describe
a fast computational algorithm.

\subsection{Time-Warping Stability}
\label{sec:stability}
Stability to time-warping allows one to use linear operators for calculating
invariant descriptors to small time-warping deformations. The Fourier
transform is unstable to deformation because dilating a sinusoidal wave yields
a new sinusoidal wave of different frequency which is orthogonal to the
original one. Section \ref{sec:mfsc} explains that mel-frequency spectrograms
become stable to time-warping deformation with a frequency averaging. One can
prove that a scattering representation $\Phi(x) = Sx$ satisfies the Lipschitz
continuity condition (\ref{eq:stability}) because wavelets are stable to
time-warping \cite{stephane}. Let us write $\psi_{\lambda,\tau}
(t) = \psi_{\lambda} (t- \tau(t))$. One can verify that there exists $C > 0$
such that $\|\psi_\lambda - \psi_{\lambda,\tau} \| \leq C \|\psi_\la\| \sup_t
|\tau'(t)|$, for all $\lambda$ and all $\tau(t)$. This property is at the core
of the scattering stability to time-warping deformations.

The squared Euclidean norm of a scattering vector $S x$ is the sum 
of its coefficients squared at all orders:
\begin{eqnarray*}
	\|Sx\|^2 &=&  \sum_{m=0}^{\mm} \|S_m x\|^2 \\
& = & \sum_{m=0}^{\mm} \sum_{\la_1,\ldots,\la_m} 
\int |S_m x(t,\la_1,\ldots,\la_m)|^2\,dt\,.
\end{eqnarray*}
We consider deformations $x_\tau (t) = x(t- \tau(t))$ with $|\tau'(t)| < 1$
and $\sup_t |\tau (t)| \ll T$, which means that the maximum displacement is
small relatively to the support of $\phi$. 
One can prove that there exists a constant $C$ such that for
all $x$ and any such $\tau$ \cite{stephane}:
\begin{equation}
	\|Sx_\tau-Sx\| \leq C\,\sup_{t}|\tau'(t)|\,\|x\|~,
\end{equation}
up to second-order terms. As explained for mel-spectral decompositions,
the constant $C$ is inversely proportional to 
the octave bandwidth of wavelet filters. Over multiple scattering layers,
we get $C = C_0 (\max_m Q_m)$. 
For Morlet wavelets, numerical 
experiments on many examples give $C_0 \approx 2$.

\subsection{Contraction and Energy Conservation}
\label{constadfsec}

We show that a scattering transform is contractive and can preserve energy.
We denote $\|A x\|^2$ the squared Euclidean norm 
of a vector of coefficients $A x$, such as $W_m x$, $S_m x$, $U_m x$ or $S x$.
Since $S x$ is computed by cascading wavelet modulus operators
$|W_m|$, which are all contractive, it results that $S$ is also contractive:
\begin{equation}
 \|S x - S x' \| \leq \|x - x' \|~.
\end{equation}
A scattering transform is therefore stable to additive noise.

If each wavelet transform is a tight frame, that is $\alpha = 0$ in \eqref{eq:paley-littlewood2}, each $|W_m|$ preserves
the signal norm. Applying this property to 
$|W_{m+1}| U_{m} x = (S_{m} x \,,\,U_{m+1} x)$ yields
\begin{equation}
\|U_{m} x \|^2 = \| S_{m} x \|^2 + \|U_{m+1} x\|^2 ~.
\label{eq:norm_cascade}
\end{equation}
Summing these equations $0 \leq m \leq \mm$ proves that 
\begin{equation}
\| x \|^2 =  \|S x \|^2 + \|U_{\mm+1} x \|^2 ~. 
\end{equation}
Under appropriate assumptions on the mother wavelet $\psi$, one can prove that
$\|U_{\mm+1} x \|$ goes to zero as $\mm$ increases, which implies that $\|S x
\| = \|x\|$ for $\mm = \infty$ \cite{stephane}. This property comes from the
fact that the modulus of analytic wavelet coefficients computes a smooth
envelope, and hence pushes energy towards lower frequencies. By iterating on
wavelet modulus operators, the scattering transform progressively propagates
all the energy of $U_m x$ towards lower frequencies, which is captured by the
low-pass filter of scattering coefficients $S_m x = U_m x \star \phi$.

One can verify numerically that $\|U_{\mm+1} x \|$ converges to zero
exponentially when $\mm$ goes to infinity and hence that $\|S x\|$ converges
exponentially to $\|x\|$. Table \ref{table:scattering-energy} gives the
fraction of energy $\|S_m x \|^2/\|x\|^2$ absorbed by each scattering order.
Since audio signals have little energy at low frequencies, $S_0 x$ is very
small and most of the energy is absorbed by $S_1 x$ for $T = 23~\mathrm{ms}$.
This explains why mel-frequency spectrograms are typically sufficient at these
small time scales. However, as $T$ increases, a progressively larger
proportion of energy is absorbed by higher-order scattering coefficients. For
$T = 1.5~\mathrm{s}$, about $56\%$ of the signal energy is captured in $S_2
x$. Section \ref{sec:model} shows that at this time scale, important amplitude
modulation information is carried by these second-order coefficients. For $T =
1.5~\mathrm{s}$, $S_3 x$ carries $25\%$ of the signal energy. It increases as
$T$ increases, but for audio classification applications studied in this
paper, $T$ remains below $1.5~\mathrm{s}$, so these third-order coefficients
are less important than first- and second-order coefficients. We therefore
concentrate on second-order scattering representations:
\begin{equation}
\label{scatondsf0s}
S x = \Bigl(S_0 x(t)\, , \, S_1 x(t,\la_1)\,,\,S_2 x(t,\la_1,\la_2) \Bigr)_{t,\lau,\lad}~.
\end{equation}

\begin{table}
	\begin{center}
	\begin{tabular}{|c|cccc|}
		\hline
		T & $m=0$ & $m=1$ & $m=2$ & $m=3$ \\
		\hline
		$23~\mathrm{ms}$ & 0.0\% & 94.5\% & 4.8\% & 0.2\% \\
		$93~\mathrm{ms}$ & 0.0\% & 68.0\% & 29.0\% & 1.9\% \\
		$370~\mathrm{ms}$ & 0.0\% & 34.9\% & 53.3\% & 11.6\% \\
		$1.5~\mathrm{s}$ & 0.0\% & 27.7\% & 56.1\% & 24.7\% \\
		\hline
	\end{tabular}
	\end{center}
	\caption{Averaged values $\|S_m x \|^2 / \|x\|^2$ computed
for signals $x$ in the TIMIT speech
dataset \cite{timit}, as a function of order $m$ and averaging scale $T$. For $m = 1$, $S_m x$ is calculated 
by Morlet wavelets with $Q_1 = 8$, and for $m= 2,3$ by cubic spline wavelets with $Q_2 = Q_3 = 1$.}
	\label{table:scattering-energy}
\end{table}

\subsection{Fast Scattering Computation}
\label{sec:fastscatt}
Subsampling scattering vectors provide a reduced representation, which leads
to a faster implementation. Since the averaging window $\phi$ has a duration of the order of $T$, we compute scattering vectors with half-overlapping windows at $t = kT/2$ with $k \in \Z$.

We suppose that $x(t)$ has $N$ samples over each frame of duration $T$,
and is thus sampled at a rate $N/T$. 
For each time frame $t = kT/2$, the number of
first-order wavelets $\psi_{\lambda_1}$ is about 
$Q_1 \log_2 N$ so there are about $Q_1
\log_2 N$ first-order coefficients $S_1 x(t , \lau)$. We now show that the
number of non-negligible second-order coefficients $S_2 x (t , \lau,\lad)$ 
which needs to be computed is about $Q_1 Q_2 (\log_2 N)^2/2$.

The wavelet transform envelope $|x \star\psi_{\la_1}(t)|$ is a
demodulated signal having approximately the same frequency bandwidth as
$\widehat \psi_{\la_1}$. Its Fourier transform is mostly supported in the interval
$[-\la_1 Q_1^{-1}, \la_1 Q_1^{-1}]$ for $\la_1 \geq 2 \pi Q_1/T$, and in
$[-2\pi T^{-1},2\pi T^{-1}]$ for $\la_1 \leq 2 \pi Q_1/T$. If the support of 
$\widehat \psi_\lad$ centered at $\lad$ does not intersect the
frequency support of $| x \star \psi_{\lau}|$, then
\[
|| x \star \psi_{\lau}| \star \psi_\lad | \approx 0~.
\]
One can verify that non-negligible second-order coefficients satisfy
\begin{equation}
\label{fresdfn}
\la_{2} \leq \max(\la_1 Q_1^{-1}, 2 \pi T^{-1}) ~.
\end{equation}
For a fixed $t$, a direct calculation then shows that there are of the order of $Q_1 Q_2 (\log_2 N)^2/2$ second-order scattering coefficients. 
Similar reasoning extends this result to show that there are 
about $Q_1 \ldots Q_m (\log_2 N)^m / m!$ non-negligible $m$th-order 
scattering coefficients. 

To compute $S_1 x$ and $S_2 x$ we first calculate $U_1 x$ and $U_2 x$ and
average them with $\phi$. Over a time frame of duration $T$, to reduce
computations while avoiding aliasing, $|x \star \psi_\lau (t)|$ is subsampled
at a rate which is twice its bandwidth. The family of filters $\{\widehat
\psi_\lau \}_{\lau \in \Lambda_1}$ covers the whole frequency domain and
$\Lambda_1$ is chosen so that filter supports barely overlap. Over a time
frame where $x$ has $N$ samples, with the above subsampling we compute
approximately $2N$ first-order wavelet coefficients $\{ |x \star \psi_\lau
(t)| \}_{t, \lau \in \Lambda_1}$. Similarly, $||x \star \psi_\lau| \star
\psi_\lad(t)|$ is subsampled in time at a rate twice its bandwidth. Over the
same time frame, the total number of second-order wavelet coefficients for all
$t$, $\lau$ and $\lad$ stays below $2 N$. With a fast Fourier transform (FFT),
these first- and second-order wavelet modulus coefficients are computed using
$O(N \log N)$ operations. The resulting scattering coefficients $S_1
x(t,\lau)$ and $S_2 x(t,\lau,\lad)$ are also calculated with $O(N \log N)$
operations, with FFT convolutions with $\phi$.

\section{Inverse Scattering}
\label{sec:inverse}
To better understand the information carried by scattering coefficients, this
section studies a numerical inversion of the transform. Since a scattering
transform is computed by cascading wavelet modulus operators $|W_m|$, the
inversion approximately inverts each $|W_m|$ for $m < \mm$. At the maximum
depth $m = \mm$, the algorithm begins with a deconvolution, estimating
$U_{\mm} x (t)$ at all $t$ on the sampling grid of $x(t)$, from $S_{\mm} x
(kT/2) = U_{\mm} x \star \phi (kT/2)$.

Because of the subsampling, one cannot compute $U_\mm x$ from $S_\mm x$
exactly. This deconvolution is thus the main source of error. To take
advantage of the fact that $U_{\mm} x \geq 0$, the deconvolution is computed
with the Richardson-Lucy algorithm \cite{lucy:1974}, which preserves
positivity if $\phi \geq 0$. We initialize $y_0(t)$ by interpolating $S_{\mm}
x(kT/2)$ linearly on the sampling grid of $x$, which introduces error because
of aliasing. The Richardson-Lucy deconvolution iteratively computes
\begin{equation} 
y_{n+1}(t) = y_n(t) \cdot
\left[\left(\frac{y_0}{y_n \star\phi}\right)\star\widetilde\phi(t)\right]\,,
\end{equation}
with $\widetilde\phi(t) = \phi(-t)$. Since it converges to the pseudo-inverse
of the convolution operator applied to $y_0$, it blows up when $n$ increases
because of the deconvolution instability. Deconvolution algorithms thus stop
the iterations after a fixed number of iterations, which is set to $30$ in
this application. The result is then our estimate of $U_{\mm}x$.

Once an estimation of $U_\mm x$ is calculated by deconvolution, we compute an
estimate $\widetilde x$ of $x$ by inverting each $|W_m|$ for $\mm \geq m > 0$. 
The wavelet transform of a signal $x$ of size $N$ is a vector
$W x = (x\star\phi,x\star\psi_\la)_{\la\in\La}$ of about $Q N \log_2 N$
coefficients, where $Q$ is the number of wavelets $\psi_\la$
per octave. These coefficients
live in a subspace $\bf V$ of dimension $N$. To recover $W x$ from
$|W| x = (x\star\phi,|x\star\psi_\la|)_{\la\in\La}$, we search for
a vector in $\bf V$ whose modulus values are specified by $|W| x$. 
This a non-convex optimization problem. Recent convex relaxation approaches
\cite{candes,irene-daspremont} are able to compute exact solutions, but they
require too much computation and memory for audio applications. Since the main
source of errors is introduced at the deconvolution stage, one can use an
approximate but fast inversion algorithm. 
The inversion of $|W|$ is typically more stable when $|W| x$ is sparse
because there is no phase to recover if $|x\star\psi_\la|=0$.
This motivates using wavelets $\psi_{\la_m}$ which provide sparse
representations at each order $m$. 

Griffin \& Lim \cite{griffin-lim} showed that alternating projections
recovers good quality audio signals from spectrogram values, but
with large mean-square errors because the algorithm is trapped in local
minima. The same algorithm inverts $|W|$ by alternating projections
on the wavelet transform space $\bV$ and on the modulus constraints.
An estimation $\widetilde x$ of $x$ is calculated from $|W| x$, 
by initializing
$\widetilde{x}_0$ to be a Gaussian white noise. For any $n \geq 0$, 
$\widetilde{x}_{n+1}$ is computed from $\widetilde{x}_n$ by first
adjusting the modulus of its wavelet coefficients, with a non-linear projector
\begin{equation}
	\label{eq:modulus-proj}
	{z}_{\la}(t) = |x\star\psi_\la(t)|\,\frac{\widetilde{x}_n\star\psi_\la(t)}{|{\widetilde{x}_n}\star\psi_\la(t)|}~.
\end{equation}
Applying the wavelet transform pseudo-inverse \eqref{eq:reproducing-kernel} yields
\begin{equation}
\label{eq:pseudo-inv}
\widetilde{x}_{n+1} = x \star \phi \star\overline{\phi}(t)+\sum_{\la\in\La}
{\rm Real}\Big(z_\la\star\overline{\psi}_{\la}(t)\Big)~.
\end{equation}
The dual filters are defined in \eqref{eq:filtern}. One can verify that $W
\widetilde {x}_{n+1}$ is the orthogonal projection of $\{x \star \phi ,
z_\lambda \}_{\la \in \Lambda}$ in $\bV$. Numerical experiments are performed
with $n = 30$ iterations, and we set $\widetilde x = \widetilde x_n$.

When $\mm = 1$, an approximation $\widetilde{x}$ of $x$ is computed from from
$(S_0 x , S_1 x)$ by first estimating $U_1 x$ from $S_1 x = U_1 x \star \phi$
with the Richardson-Lucy deconvolution algorithm. We then compute
$\widetilde{x}$ from $S_0 x$ and this estimation of $U_1 x$ by approximately
inverting $|W_1|$ with the Griffin \& Lim algorithm. When $T$ is above
$100~\mathrm{ms}$, the deconvolution loses too much information, and audio
reconstructions obtained from first-order coefficients are crude. Figure
\ref{fig:recon}(a) shows the scalograms $\log |{x} \star \psi_\lau(t)|$ of a
speech and a music signal, and the scalograms $\log |\widetilde{x} \star
\psi_\lau(t)|$ of their approximations $\widetilde x$ from first-order
scattering coefficients.

When $\mm = 2$, the approximation $\widetilde{x}$ is calculated from $(S_0 x ,
S_1 x, S_2 x)$ by applying the deconvolution algorithm to $S_2 x = U_2 x \star
\phi$ to estimate $U_2 x$, and then by successively inverting $|W_2|$ and
$|W_1|$ with the Griffin \& Lim algorithm. Figure \ref{fig:recon}(c) shows
$\log |\widetilde{x} \star \psi_\lau(t)|$ for the same speech and music
signals. Amplitude modulations, vibratos and attacks are restored with greater
precision by incorporating second-order coefficients, yielding much better
audio quality compared to first-order reconstructions. However, even with $\mm
= 2$, reconstructions become crude for $T \geq 500~\mathrm{ms}$. Indeed, the
number of second-order scattering coefficients $Q_1 Q_2 \log_2^2 N / 2$ is too
small relatively to the number $N$ audio samples in each audio frame, and they
do not capture enough information. Examples of audio reconstructions are
available at {\it http://www.di.ens.fr/data/scattering/audio/}.

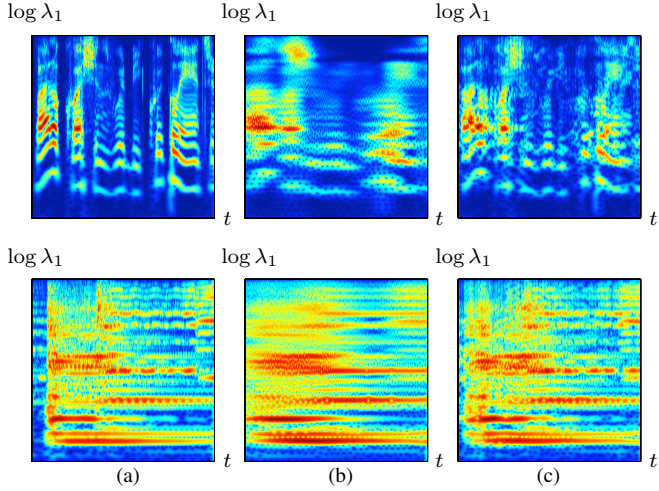
\begin{figure}[t]
\setlength{\unitlength}{1in}
\centering
\begin{picture}(3.4,2.60)
\put(0.1,0.1){\includegraphics[width=3.25in]{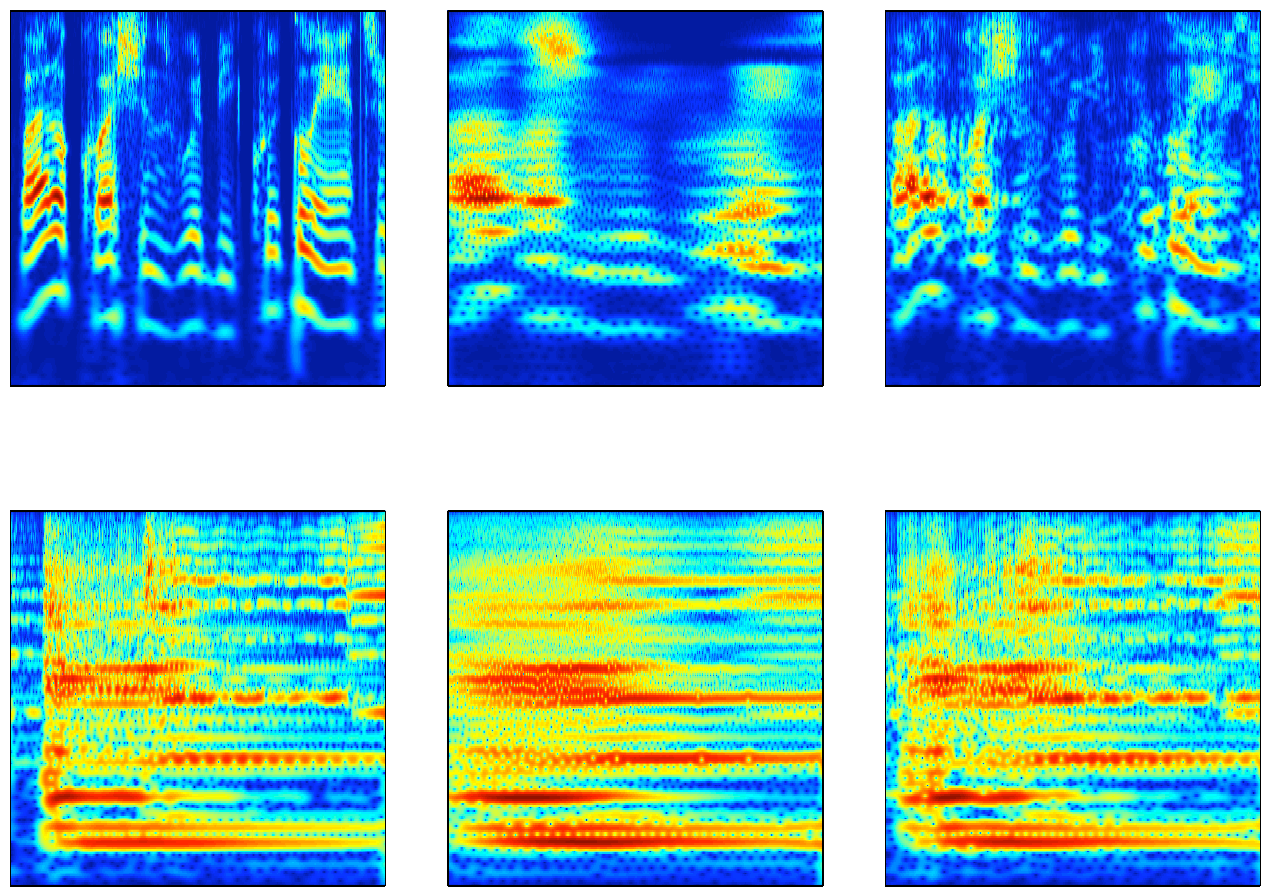}}
\put(1.13,0.08){\footnotesize $t$}
\put(2.24,0.08){\footnotesize $t$}
\put(3.36,0.08){\footnotesize $t$}
\put(0.00,1.15){\footnotesize $\log \lau$}
\put(1.12,1.15){\footnotesize $\log \lau$}
\put(2.24,1.15){\footnotesize $\log \lau$}
\put(1.13,1.35){\footnotesize $t$}
\put(2.24,1.35){\footnotesize $t$}
\put(3.36,1.35){\footnotesize $t$}
\put(0.00,2.43){\footnotesize $\log \lau$}
\put(1.12,2.43){\footnotesize $\log \lau$}
\put(2.24,2.43){\footnotesize $\log \lau$}
\put(0.57,0){\footnotesize (a)}
\put(1.68,0){\footnotesize (b)}
\put(2.77,0){\footnotesize (c)}
\end{picture}
\caption{\label{fig:recon}{\it (a): Scalogram $\log |x \star \psi_\lau(t)|$ for recordings of speech (top) and a cello (bottom).
(b,c): Scalograms $\log |\widetilde{x} \star \psi_\lau(t)|$ of reconstructions $\widetilde x$ from first-order scattering coefficients ($\mm = 1$) in (b), and from first- and  second-order coefficients ($\mm = 2$) in (c). 
Scattering coefficients were computed with 
$T = 190~\mathrm{ms}$ for the speech signal and $T = 370~\mathrm{ms}$ for the cello signal.}}
\end{figure}

\section{Normalized Scattering Spectrum}
\label{sec:model}

To reduce redundancy and increase invariance, Section \ref{sec:norm}
normalizes scattering coefficients. Section \ref{interfer} shows that
normalized second-order coefficients provide high-resolution spectral
information through interferences. Section \ref{ModusSpec} also proves that
they characterize amplitude modulations of audio signals.

\subsection{Normalized Scattering Transform}
\label{sec:norm}

Scattering coefficients are renormalized to increase their invariance.
It also decorrelates these coefficients at different orders.
First-order scattering coefficients are renormalized so that they become
insensitive to multiplicative constants:
\begin{equation}
	\label{eq:normscatt01}
\tS_1 x(t,\lau) = \frac{S_1x(t,\lau)}{|x|\star\phi(t) + \epsilon}~.
\end{equation}
The constant $\epsilon$ is a silence detection threshold so that $\tS_1 x = 0$
if $x = 0$, which may be set to $0$.

The lowpass filter $\phi(t)$ can be wider than the one used in the scattering transform. Specifically, if we want to retain local amplitude information of $S_1x$ below a certain scale, we can normalize by the average of $|x|$ over this scale, creating invariance only to amplitude changes over larger intervals.

At any order $m \geq 2$, scattering coefficients are renormalized
by coefficients of the previous order:
\[
	\tS_m x(t,\lau,...,\la_{m-1},\la_m) 
= \frac{S_m x(t,\lau,...,\la_{m-1},\la_{m})}{S_{m-1} x(t,\lau,...,\la_{m-1})+\epsilon}~.
\]
A normalized scattering representation is defined by 
$\widetilde S x= (\widetilde S_m x)_{1 \leq m \leq \mm}$. We shall mostly limit ourselves
to $\mm = 2$.

For $m = 2$, 
\begin{equation}
	\tS_2 x(t,\lau,\lad) =  \frac{S_2 x(t,\lau,\lad)}{S_1 x(t,\lau)+\epsilon}~.
\end{equation}
Let us show that these coefficients are nearly invariant to a filtering by
$h(t)$ if $\widehat{h}(\om)$ is approximately
constant on the support of $\widehat{\psi}_\lau$. This condition is satisfied if
\begin{equation}
\label{ocnionsdf09sdf}
\frac {\la_1}  {Q_1} \ll \left( \int |t|\,|h(t)|\,dt\right)^{-1} ~.
\end{equation}
It implies that 
$h\star\psi_\lau(t) \approx \widehat{h}(\lau)\psi_\lau(t)$, and hence 
$|(x\star{}h)\star\psi_\lau(t)| \approx
|\widehat{h}(\lau)|\,|x\star\psi_\lau(t)|$. It results that
\begin{equation}
	S_1 (x\star{}h)(t,\lau) \approx |\widehat{h}(\lau)| S_1x(t,\lau).
\end{equation}
Similarly, 
$S_2 (x\star{}h)(t,\lau,\lad) \approx |\widehat{h}(\lau)| S_2x(t,\lau,\lad)$, so
after normalization
\begin{equation}
\tS_2 (x\star{}h)(t,\lau,\lad) \approx \tS_2 x(t,\lau,\lad)~.	
\end{equation}
Normalized second-order coefficients are thus invariant to 
filtering by $h(t)$. One can verify that this remains valid
at any order $m \geq 2$.

\subsection{Frequency Interval Measurement from Interference}
\label{interfer}
A wavelet transform has a worse frequency resolution than a windowed Fourier
transform at high frequencies. However, we show that frequency intervals
between harmonics are accurately measured by second-order scattering
coefficients.

Suppose $x$ has two frequency components in the support of
$\widehat \psi_\lau$. We then have
\[
x \star \psi_{\lau} (t) = \alpha_1 \, e^{i \xi_1 t} + \alpha_2 \, e^{i \xi_2 t},
\] 
whose modulus squared equals
\[
	|x\star\psi_\lau(t)|^2 = |\alpha_1|^2+|\alpha_2|^2+2|\alpha_1\alpha_2|\cos(\xi_1-\xi_2)t.
\]
We approximate $|x\star\psi_\lau(t)|$ with a first-order expansion of the square root, which yields
\[
	|x\star\psi_\lau(t)| \approx \sqrt{|\alpha_1|^2+|\alpha_2|^2}+\frac{|\alpha_1\alpha_2|}{\sqrt{|\alpha_1|^2+|\alpha_2|^2}}\cos(\xi_1-\xi_2)t.
\]
If $\phi$ has a support of size $T \gg |\xi_1-\xi_2|^{-1}$, then
$S_1 x(t,\lau) \approx \sqrt{|\alpha_1|^2+|\alpha_2|^2}$, so
$\tS_2 x(t,\lau,\lad) = \frac{S_2x (t,\lau,\lad)}{S_1x(t,\lau)+\epsilon}$ satisfies
\begin{equation} \label{nosdfn8sdfssd8} 
\tS_2 x(t,\lau,\lad) 
\approx |\widehat \psi_\lad (\xi_2 - \xi_1)| \,
\frac{|\alpha_1\, \alpha_2|} {|\alpha_1|^2 + |\alpha_2|^2}~. \end{equation}
These normalized second-order coefficients are thus non-negligible when $\lad$
is of the order of the frequency interval $|\xi_2-\xi_1|$. This shows that
although the first wavelet $\widehat \psi_\lau$ does not have enough
resolution to discriminate the frequencies $\xi_1$ and $\xi_2$, second-order
coefficients detect their presence and accurately measure the interval $|\xi_2
- \xi_1|$. As in audio perception, scattering coefficients can accurately
measure frequency intervals but not frequency location. The normalized
second-order scattering coefficients (\ref{nosdfn8sdfssd8}) are large only if
$\alpha_1$ and $\alpha_2$ have the same order of magnitude. This also conforms
to auditory perception where a frequency interval is perceived only when the
two frequency components have a comparable amplitude.

If $x \star \psi_{\lau}(t) = \sum_n \alpha_n \, e^{i \xi_n t}$ has more
frequency components, we verify similarly that $\tS_2 x(t,\lau,\lad)$ is
non-negligible when $\lad$ is of the order of $|\xi_n - \xi_{n'}|$ for some $n
\neq n'$. These coefficients can thus measure multiple frequency intervals
within the frequency band covered by $\widehat \psi_{\lau}$. If the frequency
resolution of $\widehat \psi_{\lad}$ is not sufficient to discriminate between
two frequency intervals $|\xi_{1} - \xi_{2}|$ and $|\xi_{3} - \xi_{4}|$, these
intervals will interfere and create high amplitude third-order scattering
coefficients. A similar calculation shows that third-order scattering
coefficients $\tS_3 x(t,\lau,\lad,\la_3)$ detect the presence of two such
intervals within the support of $\widehat \psi_{\lad}$ when $\la_3$ is close
to $||\xi_{1} - \xi_{2}| - |\xi_{3} - \xi_{4}||$. They thus measure
``intervals of intervals.''

Figure \ref{fig:chord-display}(a) shows the scalogram $\log |x\star\psi_\lau|$
of a signal $x$ containing a chord with two notes, whose fundamental
frequencies are $\xi_1 = 600~\mathrm{Hz}$ and $\xi_2 = 675~\mathrm{Hz}$,
followed by an arpeggio of the same two notes. First-order coefficients $\log
\tS_1 x(t,\lau)$ in Figure \ref{fig:chord-display}(b) are very similar for the
chord and the arpeggio because the time averaging loses time localization.
However they are easily differentiated in Figure \ref{fig:chord-display}(c),
which displays $\log \tS_2x (t,\lau,\lad)$ for $\lau \approx \xi_1 =
600~\mathrm{Hz}$, as a function of $\lad$. The chord creates large amplitude
coefficients for $\lad = |\xi_2 - \xi_1| = 75~\mathrm{Hz}$, which disappear for
the arpeggio because these two frequencies are not present simultaneously.
Second-order coefficients have also a large amplitude at low frequencies
$\la_2$. These arise from variation of the note envelopes in the chord and in
the arpeggio, as explained in the next section.

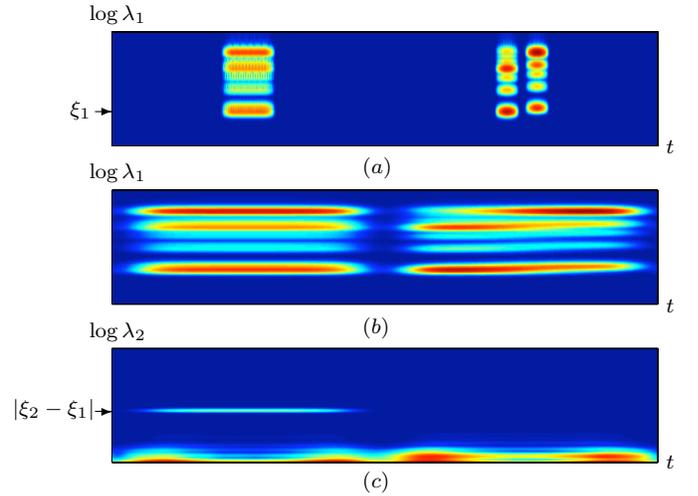
\begin{figure}[t]
\center
\setlength{\unitlength}{1in}
\begin{picture}(3.5,2.55)
\put(0.5,0.13){\includegraphics[width=2.9in]{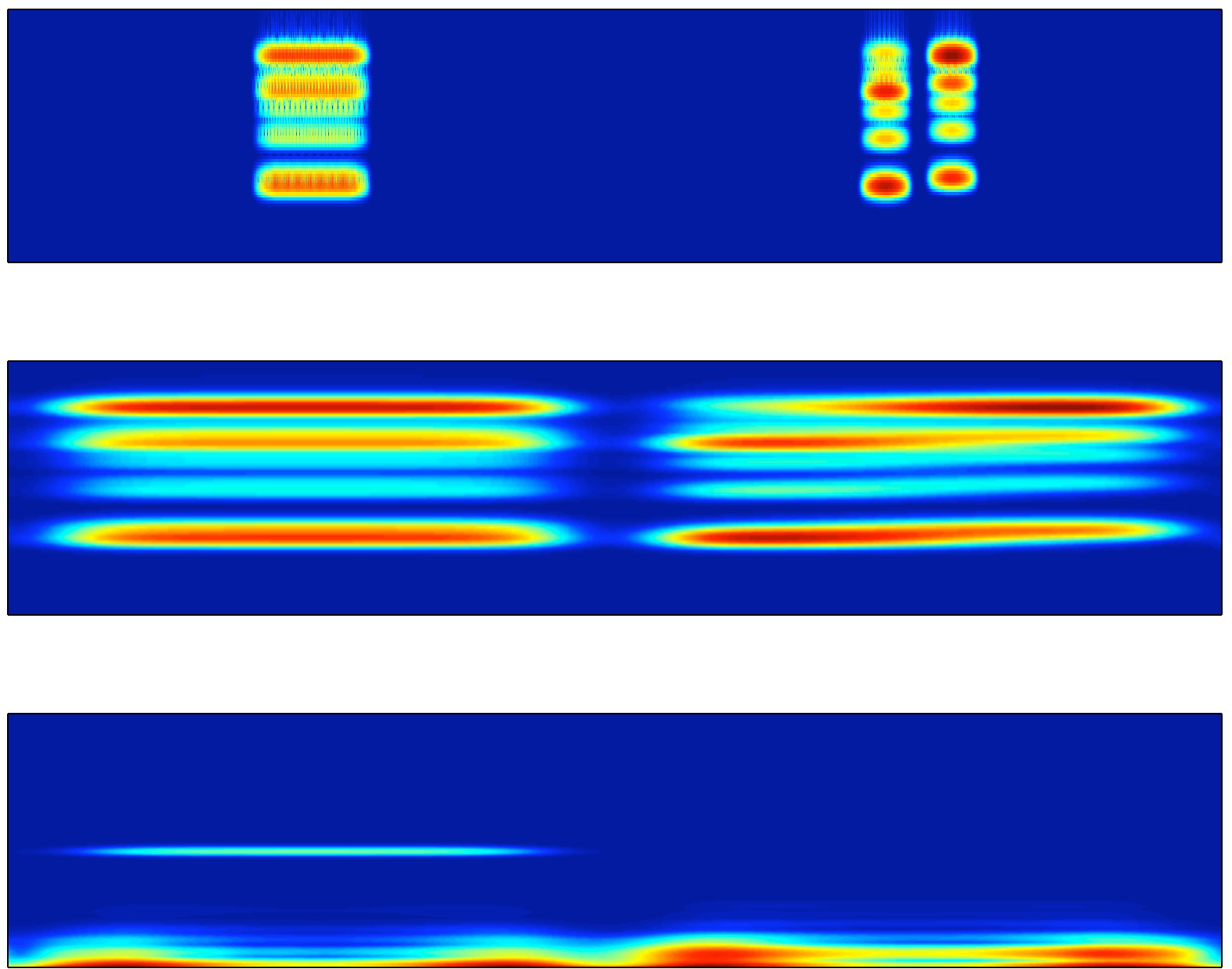}}
\put(3.42,0.10){\footnotesize $t$}
\put(3.42,0.92){\footnotesize $t$}
\put(3.42,1.75){\footnotesize $t$}
\put(0.4,0.79){\footnotesize $\log\lad$}
\put(0.4,1.62){\footnotesize $\log\lau$}
\put(0.4,2.45){\footnotesize $\log\lau$}
\put(1.83,0){\footnotesize $(c)$}
\put(1.83,0.81){\footnotesize $(b)$}
\put(1.83,1.65){\footnotesize $(a)$}
\put(0.43,1.97){\vector(+1,0){0.09}}
\put(0.30,1.95){\footnotesize $\xi_1$}
\put(0.43,0.40){\vector(+1,0){0.09}}
\put(0,0.38){\footnotesize $|\xi_2-\xi_1|$}
\end{picture}
\caption{\label{fig:chord-display}{\it (a): Scalogram $\log|x \star
\psi_{\lau}(t)|$ for a signal with two notes, of fundamental frequencies
$\xi_1 = 600~\mathrm{Hz}$ and $\xi_2 = 675~\mathrm{Hz}$, first played as a
chord and then as an arpeggio. (b): First-order 
normalized scattering coefficients $\log
\tS_1 x(t,\lau)$ for $T = 512~\mathrm{ms}$. (c): Second-order 
normalized scattering
coefficients $\log\tS_2(t,\xi_1,\lad)$ with $\lau = \xi_1$ as a function of $t$ and
$\la_2$. The chord interferences produce large coefficients for $\lad = |\xi_2
- \xi_1|$. }}
\end{figure}

\subsection{Amplitude Modulation Spectrum}
\label{ModusSpec}
Audio signals are usually modulated in amplitude by an envelope, 
whose variations may correspond to an attack or a tremolo.
For voiced and unvoiced sounds,
we show that amplitude modulations
are characterized by normalized second-order scattering coefficients.

Let $x(t)$ be a sound resulting from an
excitation $e(t)$ filtered by a resonance cavity of impulse response $h(t)$,
which is modulated in amplitude by $a(t) \geq 0$ to give
\begin{equation}
	\label{eq:source-filter-modulation}
	x(t) = a(t)\,(e\star{}h)(t)~.
\end{equation}

We shall start by taking $e$ to be a pulse train of pitch $\xi$ given by
\begin{equation}
	\label{eq:pitched-excitation}
	e(t) = \frac{2\pi}{\xi}\, \sum_n \delta \left(t - \frac{2 n \pi}\xi \right) = \sum_k e^{i k \xi t}~,
\end{equation}
representing a voiced sound. 
The impulse response $h(t)$ is typically very short compared to the minimum
variation interval $(\sup_t |a'(t)|)^{-1}$ of the modulation term and is smaller
than $2\pi/\xi$. 

We consider $\psi_{\lau}$ whose time support is short relatively
to $(\sup_t |a'(t)|)^{-1}$ and to the averaging interval $T$,
and whose frequency bandwidth is smaller than the
pitch $\xi$ and to the minimum variation interval of $\hat h$. 
These conditions are satisfied if
\begin{equation}
	\label{eq:formant-separation10sdf9d}
\left(\int |t|\,|h(t)|\,dt\right)^{-1} \gg \frac{\lau} {Q_1} \gg
\sup_t |a'(t)|~,
\end{equation}

After normalization
$\tS_1x(t,\lau) = \frac {S_1 x(t,\lau)} {|x|\star\phi(t)+\epsilon}$,
Appendix \ref{append} shows that
\begin{equation}
	\label{eq:first-order-factorization}
\tS_1x(t,\lau) 
\approx |\widehat \psi_\lau (k \xi)|\frac{|\widehat{h}(\lau)|}{\|h\|_1}
\end{equation}
where $\|h\|_1 (t) = \int |h(t)|dt$ and 
$k$ is an integer such that $|k \xi - \lau| < \xi/2$.
First-order coefficients
are thus proportional to the spectral
envelope $|\widehat{h}(\lau)|$ if $\lambda_1 \approx k \xi$ is close to a harmonic
frequency. 

Similarly, for $\tS_2x(t,\lau,\lad) = \frac{S_2 x (t,\lau,\lad)}
{S_1 x(t,\lau)+\epsilon} $, Appendix \ref{append} shows that
\begin{equation}
	\label{eq:first-order-factorization3}
\tS_2x(t,\lau,\lad) \approx \frac{|a\star\psi_\lad| \star \phi (t)}
{a\star \phi (t)}~.
\end{equation}
Second-order coefficients thus do not depend upon $h$ and $\xi$ but only
on the amplitude modulation $a(t)$ provided that $S_1 x(t,\lau)$ is non-negligible.

Figure \ref{fig:am-display}(a) displays $\log|x \star \psi_{\lau}(t)|$ for a
signal having three voiced and three unvoiced sounds. The first three are
produced by a pulse train excitation $e(t)$ with a pitch of $\xi =
600~\mathrm{Hz}$. Figure \ref{fig:am-display}(b) shows that $\log \tS_1
x(t,\lau)$ has a harmonic structure, with an amplitude depending on $\log
|\widehat h(\lau)|$. The averaging by $\phi$ and the normalization remove the
effect of the different modulation amplitudes $a(t)$ of these three voiced
sounds.

\begin{figure}[t]
\center
\setlength{\unitlength}{1in}
\begin{picture}(3.2,2.55)
\put(0.2,0.13){\includegraphics[width=2.9in]{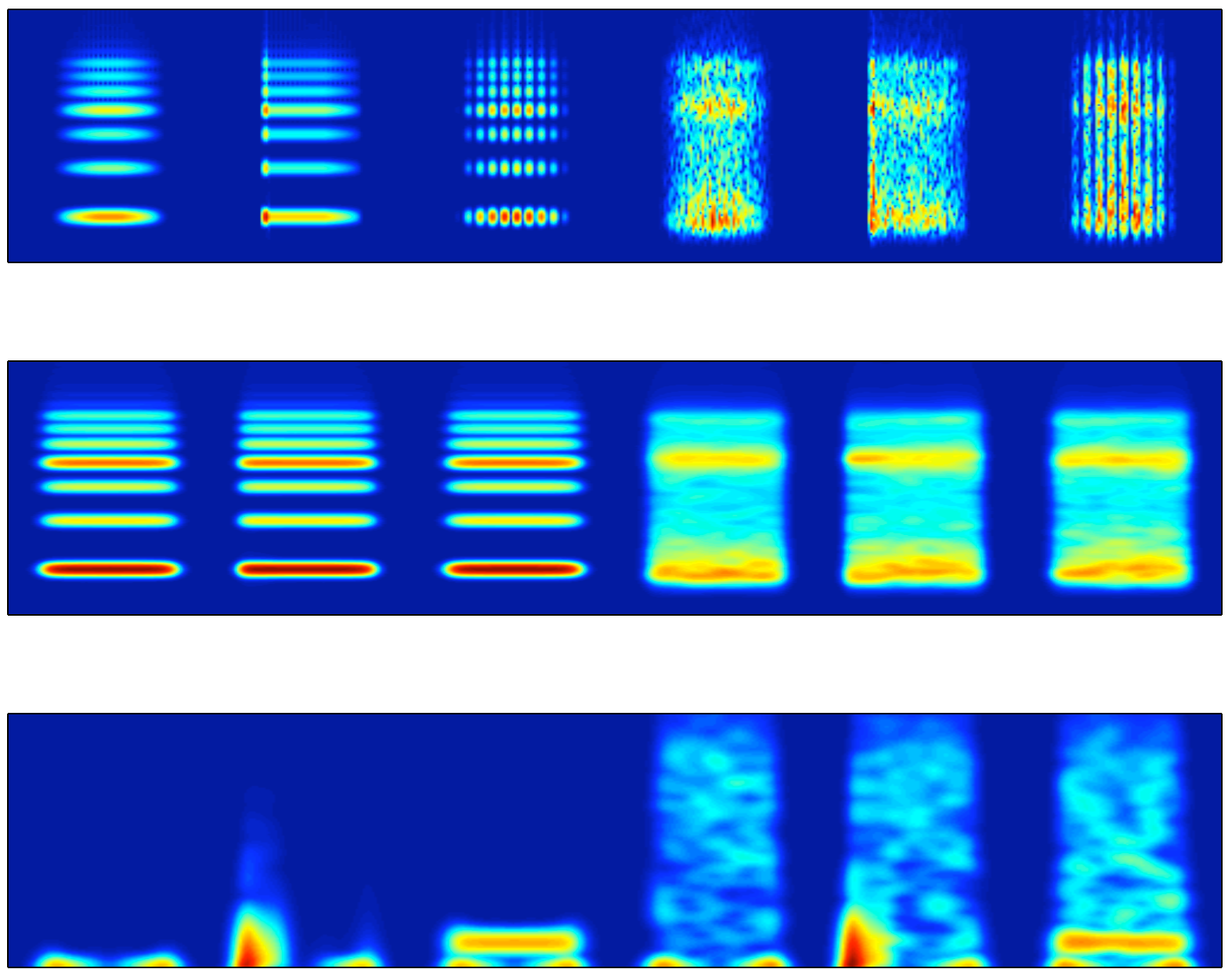}}
\put(3.12,0.10){\footnotesize $t$}
\put(3.12,0.92){\footnotesize $t$}
\put(3.12,1.75){\footnotesize $t$}
\put(0.1,0.79){\footnotesize $\log\lad$}
\put(0.1,1.62){\footnotesize $\log\lau$}
\put(0.1,2.45){\footnotesize $\log\lau$}
\put(1.57,0){\footnotesize $(c)$}
\put(1.57,0.81){\footnotesize $(b)$}
\put(1.57,1.64){\footnotesize $(a)$}
\put(0.13,2.15){\vector(+1,0){0.09}}
\put(0.0,2.14){\footnotesize $4\xi$}
\put(0.13,0.20){\vector(+1,0){0.09}}
\put(0.05,0.18){\footnotesize $\eta$}
\end{picture}

\caption{\label{fig:am-display}{
\it (a): Scalogram $\log|x \star \psi_{\lau}(t)|$ for a signal with three 
voiced sounds of same pitch $\xi = 600~\mathrm{Hz}$ and same $h(t)$
but different amplitude modulations $a(t)$: first a smooth attack,
then a sharp attack, then a tremolo of frequency $\eta$. It is  
followed by three unvoiced sounds created with the same $h(t)$ and 
same amplitude modulations $a(t)$ as the first three voiced sounds.
(b): First-order scattering $\log \tS_1 x(t,\lau)$ with $T = 128~\mathrm{ms}$. 
(c): Second-order scattering $\log \tS_2 x(t,\lau,\lad)$ displayed 
for $\lau = 4 \xi$, as a function of $t$ and $\lad$.
}}
\end{figure}
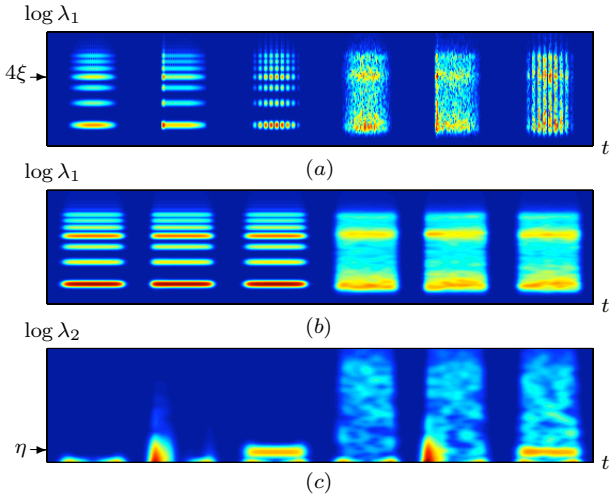

Figure \ref{fig:am-display}(c) displays $\log \tS_2(t,\lau,\lad)$
for the fourth partial $\lau = 4 \xi$, as a function of $\la_2$. The
modulation envelope $a(t)$ of the first sound has a smooth attack and thus
produces large coefficients only at low frequencies $\lad$. The envelope
$a(t)$ of the second sound has a much sharper attack and thus produces large
amplitude coefficients for higher frequencies $\lad$. The third sound is
modulated by a tremolo, which is a periodic oscillation $a(t) =
1+\epsilon\cos{(\eta t)}$. According to (\ref{eq:first-order-factorization3}),
this tremolo creates large amplitude coefficients when $\lad = \eta$, as shown
in Figure \ref{fig:am-display}(c).

Unvoiced sounds are modeled by excitations $e(t)$ which are realizations of
Gaussian white noise. The modulation amplitude is typically non-sparse, which
means the square of the average of $a(t)$ on intervals of size $T$ is of the
order of the average of $a^2(t)$. Appendix \ref{append} shows that
\begin{equation}
	\label{eq:first-order-factorization0}
	\tS_1 x(t,\lau) 
\approx \frac{\pi \|\psi\|}{2^{3/2}}\,\lau^{1/2}\,
\frac {|\widehat h (\lau)|}{\|h\|}~.
\end{equation}
Similarly to (\ref{eq:first-order-factorization}), $\tS_1 x(t,\lau)$
is proportional to $|\widehat h(\lau)|$ but does not have a harmonic structure. This is shown in Figure
\ref{fig:am-display}(b) by the last three unvoiced sounds. The fourth, fifth,
and sixth sounds have the same filter $h(t)$ and envelope $a(t)$ as the first,
second, and third sounds, respectively, but with a Gaussian white noise
excitation $e(t)$.

Similarly to (\ref{eq:first-order-factorization3}),
Appendix \ref{append} also shows that 
\[
\tS_2x(t,\lau,\lad) = 
\frac{|a\star\psi_\lad| \star \phi (t)}
{a\star \phi (t)} + \widetilde \epsilon (t)
\]
where $\widetilde \epsilon (t)$ is small relatively to the first amplitude
modulation term if $(4/\pi-1)^{1/2}(\lad Q_1)^{1/2}(\lau Q_2)^{-1/2}$ is small
relatively to this modulation term. 
For voiced and unvoiced sounds, $\tS_2x(t,\lau,\lad)$ 
mainly depends upon the
amplitude modulation $a(t)$. This is illustrated by Figure
\ref{fig:am-display}(c), which shows that the fourth, fifth, and sixth sounds
have second-order coefficients similar to those of the first, second, and
third sounds, respectively. The stochastic error term $\widetilde \epsilon$
produced by unvoiced sounds appears as random low-amplitude fluctuations in
Figure \ref{fig:am-display}(c).

\section{Frequency Transposition Invariance}
\label{sec:cls}
Audio signals within the same class may be transposed in frequency. Frequency
transposition occurs when a single word is pronounced by different speakers.
It is a complex phenomenon which affects the pitch and the spectral envelope.
The envelope is translated on a logarithmic frequency scale but also deformed.
We thus need a representation which is invariant to frequency translation on a
logarithmic scale, and which also is stable to frequency deformations. After
reviewing the mel-frequency cepstral coefficient (MFCC) approach through the
discrete cosine transform (DCT), this section defines such a representation
with a scattering transform computed along log-frequency.

MFCCs are computed from the log-mel-frequency spectrogram $\log M
x(t,\lambda)$ by calculating a DCT along the mel-frequency index $\gamma$ for
a fixed $t$ \cite{davis-mermelstein}. This $\gamma$ is linear in $\lambda$ for
low frequencies, but is proportional to $\log_2\lambda$ for higher
frequencies. For simplicity, we write $\gamma = \log_2\lambda$
and $\lambda = 2^\gamma$, although this should be modified at low frequencies.

The frequency index of the DCT is called the ``quefrency'' parameter. In
MFCCs, high-quefrency coefficients are often set to zero, which is equivalent
to averaging $\log Mx(t,2^\gamma)$ along $\gamma$ and provides some
frequency transposition invariance. The more high-quefrency coefficients are
set to zero, the bigger the averaging and hence the more transposition
invariance obtained, but at the expense of losing potentially important
information.

The loss of information due to averaging along $\gamma$ can be recovered by
computing wavelet coefficients along $\gamma$. We thus replace the DCT by a
scattering transform along $\gamma$. A frequency scattering transform is
calculated by iteratively applying wavelet transforms and modulus operators.
An analytic wavelet transform of a log-frequency dependent signal $z(\gamma)$
is defined as in (\ref{wavensdfons}), but with convolutions along the
log-frequency variable $\gamma$ instead of time:
\begin{equation}
\label{wavensdfons2}
	 W^\freq z = 
\Bigl(z\star \phi^\freq(\gamma)\,,\, z \star \psi_\quef (\gamma)
\Bigr)_{\gamma,\quef}~.
\end{equation}
Each wavelet $\psi_\quef$ is a band-pass filter whose Fourier transform
$\widehat{\psi}_\quef$ is centered at ``quefrency'' $\quef$ and $\phi^\freq$
is an averaging filter. These wavelets satisfy the condition
(\ref{eq:paley-littlewood2}), so $W^\freq$ is contractive and invertible.

Although the scattering transform along $\gamma$
can be computed at any order,
we restrict ourself to zero and
first-order scattering coefficients, because it 
seems to be sufficient for classification. 
A first-order scattering transform of $z(\gamma)$ is calculated from 
\begin{equation}
\label{Ufreq}
\Ufreq z = 
\Big(z (\gamma)\,,\,|z \star \psi_\quefu (\gamma)|\Big)\,,
\end{equation}
by averaging these coefficients along $\gamma$ with $\phi^\freq$:
\begin{equation}
\label{Sfreq}
\Sfreq z = 
\Big(z \star \phi^\freq(\gamma)\,,\,|z \star \psi_\quefu | \star \phi^\freq(\gamma) \Big)\,.
\end{equation}
These coefficients are locally invariant to log-frequency shifts, over a
domain proportional to the support of the averaging filter $\phi^\freq$. This
frequency scattering is formally identical to a time scattering transform. It
has the same properties if we replace the time $t$ by the log-frequency
variable $\gamma$. Numerical experiments are implemented using Morlet wavelets
$\psi_{q_1}$ with $Q_1 = 1$.

Similarly to MFCCs, we apply a logarithm to normalized scattering coefficients
so that multiplicative components become additive and can be separated by
linear operators. This was shown to improve classification performance. The
logarithm of a second-order normalized time scattering, at a frequency $\lau =
2^\gamma$ and a time $t$ is
\begin{equation}
\label{eq:log-scatt-j}
\log \tS x(t,\gamma) =
\left(
\begin{array}{l}
\log \tS_1 x(t,2^\gamma)\\
\log \tS_2 x(t,2^\gamma,\lad)
\end{array}
\right)_{\lad}
\end{equation}
This is a vector of signals $z(\gamma)$, where $z$ depends on $t$ and $\la_2$.
Let us transform each $z(\gamma)$ by the frequency scattering operators
$\Ufreq$ or $\Sfreq$, defined in (\ref{Ufreq}) and (\ref{Sfreq}). Let $\Ufreq
\log \tS x(t,\gamma)$ and $\Sfreq \log \tS x(t,\gamma)$ stand for the
concatenation of these transformed signals for all $t$ and $\la_2$. The
representation $\Sfreq \log \tS x$ is calculated by cascading a scattering in
time and a scattering in log-frequency. It is thus locally translation
invariant in time and in log-frequency, and stable to time and frequency
deformations. The interval of time-shift invariance is defined by the size of
the time averaging window $\phi$, whereas its frequency-transposition
invariance depends upon the width of the log-frequency averaging window
$\phi^\freq$.

Frequency transposition invariance is useful for certain tasks, such as speaker-independent speech
recognition or transposition-independent melody recognition, but it
removes information important to other tasks, such as speaker identification.
The frequency transposition invariance, implemented by the frequency averaging
filter $\phi^\freq$, should thus be adapted to the classification task. Next
section explains that this can be done by replacing $\Sfreq \log \tS
x(t,\gamma)$ by $\Ufreq \log \tS x(t,\gamma)$ and optimizing the linear
averaging at the supervised classification stage.

\section{Classification}
\label{sec:classif}

This section compares the classification performance of support vector machine
classifiers applied to scattering representations with standard low-level
features such as $\Delta$-MFCCs or more sophisticated state-of-the-art
representations. Section \ref{sec:arch} explains how to automatically adapt
invariance parameters, while Sections \ref{sec:gtzan} and \ref{sec:phone}
present results for musical genre classification and phone classification,
respectively.

\subsection{Adapting Frequency Transposition Invariance}
\label{sec:arch}

\begin{figure}
\centering
\includegraphics[width=2.5in]{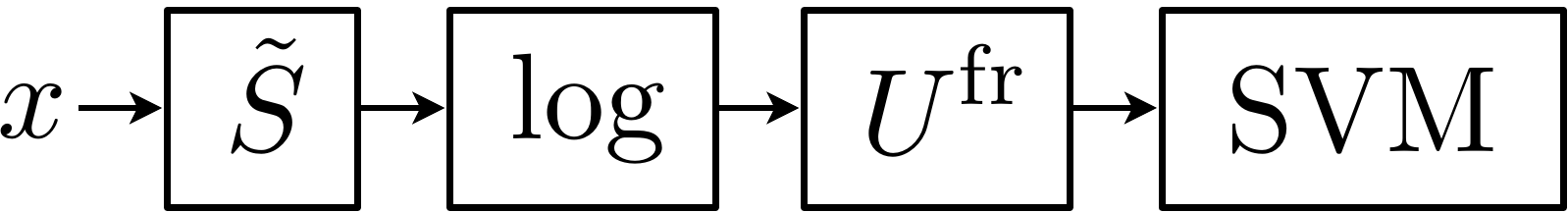} \caption{\label{fig:arch}{\it A time and frequency scattering representation is computed by applying a
normalized temporal scattering $\tS$ on the input signal $x(t)$, a logarithm,
and a scattering along log-frequency without averaging.}}
\end{figure}

The amount of frequency-transposition invariance depends on the
classification problem, and may vary for each signal class. This adaptation is
implemented by a supervised classifier, applied to the time and frequency
scattering transform.

Figure \ref{fig:arch} illustrates the computation of a time and frequency
scattering representation. The normalized scattering transform $\tS x$ of an
input signal $x$ is computed along time, over half-overlapping windows of size
$T$. The log-scattering vector for each time window is transformed along
frequencies by the wavelet modulus operator $\Ufreq$, as explained in Section
\ref{sec:cls}. Since we do not know in advance how much transposition
invariance is needed for a particular classification task, the final frequency
averaging is adaptively computed by the supervised classifier, which takes as
input the vector of coefficients $\{ \Ufreq \log \tS x(t,\gamma) \}_\gamma$,
for each time frame indexed by $t$.

The supervised classification is implemented by a support vector machine
(SVM). A binary SVM classifies a feature vector by calculating its position
relative to a hyperplane, which is optimized to maximize class separation
given a set of training samples. It thus computes the sign of an optimized
linear combination of the feature vector coefficients. With a Gaussian kernel
of variance $\sigma^2$, the SVM computes different hyperplanes in different
balls of radius $\sigma$ in the feature space. The coefficients of the linear
combination thus vary smoothly with the feature vector values. Applied to $\{
\Ufreq \log \tS x(t,\gamma) \}_\gamma$, the SVM optimizes the linear
combination of coefficients along $\gamma$, and can thus adjust the amount of
linear averaging to create frequency-transposition invariant descriptors which
maximize class separation. A multi-class SVM is computed from binary
classifiers using a one-versus-one approach. All numerical experiments use the
LIBSVM library \cite{libsvm}.

The wavelet octave resolution $Q_1$ can also be adjusted at the supervised
classification stage, by computing the time scattering for several values of
$Q_1$ and concatenating all coefficients in a single feature vector. A filter
bank with $Q_1 = 8$ has enough frequency resolution to separate harmonic
structures, whereas wavelets with $Q_1 = 1$ have a smaller time support and
can thus better localize transient in time. The linear combination optimized
by the SVM is a feature selection algorithm, which can select the best
coefficients to discriminate any two classes. In the experiments described
below, adding more values of $Q_1$ between $1$ and $8$ provides marginal
improvements.

\subsection{Musical Genre Classification}
\label{sec:gtzan}

Scattering feature vectors are first applied to musical genre classification
problem on the GTZAN dataset \cite{gtzan}. The dataset consists of $1000$
thirty-second clips, divided into $10$ genres of $100$ clips each.
Given a clip, the goal is to find its genre.

Preliminary experiments have demonstrated the efficiency of the scattering
transform for music classification \cite{ismir} and for environmental sounds
\cite{bauge2013}. These results are improved by letting the supervised
classifier adjust the transform parameters to the signal classes. A set of
feature vectors is computed over half-overlapping frames of duration $T$. Each
frame of a clip is classified separately by a Gaussian kernel SVM, and the
clip is assigned to the class which is most often selected by its frames. To
reduce the SVM training time, feature vectors were only computed every
$370~\mathrm{ms}$ for the training set. The SVM slack parameter and the
Gaussian kernel variance are determined through cross-validation on the
training data. Table \ref{table:results} summarizes results with one run of
ten-fold cross-validation. It gives the average error and its standard
deviation.

Scattering classification results are first compared with results obtained
with MFCC feature vectors. A $\Delta$-MFCC vector represents an audio frame of
duration $T$ at time $t$ by three MFCC vectors centered at $t-T/2$, $t$ and
$t+T/2$. When computed for $T=23~\mathrm{ms}$, the $\Delta$-MFCC error is
$20.2\%$, which is reduced to $18.0\%$ by increasing $T$ to $740~\mathrm{ms}$.
Further increasing $T$ does not reduce the error. State-of-the-art algorithms
provide refined feature vectors to improve classification. Combining MFCCs
with stabilized modulation spectra and performing linear discriminant
analysis, \cite{lee-shih} obtains an error of $9.4\%$, the best non-scattering
result so far. A deep belief network trained on spectrograms
\cite{hamel2010learning}, achieves $15.7\%$ error with an SVM classifier. A
sparse representation on a constant-Q transform \cite{henaff2011unsupervised},
gives $16.6\%$ error with an SVM.

\begin{table}
	\begin{center}
	\begin{tabular}{|l|c|c|}
		\hline
		Representations & GTZAN & TIMIT \\
		\hline
		$\Delta$-MFCC (T = $23~\mathrm{ms}$) & 20.2 $\pm$ 5.4 & 18.5 \\
		$\Delta$-MFCC (T = $740~\mathrm{ms}$) & 18.0 $\pm$ 4.2 & 60.5 \\
		State of the art (excluding scattering) & 9.4 $\pm$ 3.1 \cite{lee-shih} & 16.7 \cite{chang-glass} \\
		\hline
		& $T = 740~\mathrm{ms}$ & $T = 32~\mathrm{ms}$ \\
		\hline
		Time Scat., $\mm=1$ & 19.1 $\pm$ 4.5 & 19.0 \\
		Time Scat., $\mm=2$ & 10.7 $\pm$ 3.1 & 17.3 \\
		Time Scat., $\mm=3$ & 10.6 $\pm$ 2.5 & 18.1 \\
		Time \& Freq. Scat., $\mm=2$ & 9.3 $\pm$ 2.4 & 16.6 \\
		Adapt $Q_1$, Time \& Freq. Scat., $\mm=2$ & 8.6 $\pm$ 2.2 & 15.9 \\
		\hline
	\end{tabular}
	\end{center}
	\caption{Error rates (in percent) for musical genre classification on GTZAN and for phone classification on the TIMIT database for different features. Time scattering transforms are computed with $T =
740~\mathrm{ms}$ for GTZAN and with $T = 32~\mathrm{ms}$ for TIMIT.}
	\label{table:results}
\end{table}

Table \ref{table:results} gives classification errors for different scattering
feature vectors. For $\mm = 1$, they are composed of first-order time
scattering coefficients computed with $Q_1 =
8$ and $T = 740~\mathrm{ms}$. These vectors are similar to MFCCs as shown by
(\ref{Melwaves}). As a result, the classification error of $19.1\%$ is close
to that of MFCCs for the same $T$. For $\mm = 2$, we add second-order
coefficients computed with $Q_2 = 2$. It reduces the
error to $10.7 \%$. This $40\%$ error reduction shows the importance of
second-order coefficients for relatively large $T$. Third-order coefficients
are also computed with $Q_3 = 1$. For $\mm = 3$,
including these coefficients reduces the error marginally to $10.6\%$, at a
significant computational and memory cost. We therefore restrict ourselves to $\mm
= 2$.

Musical genre recognition is a task which is partly invariant to frequency
transposition. Incorporating a scattering along the log-frequency variable,
for frequency transposition invariance, reduces the error by about $15\%$. These errors are obtained with a first-order scattering along log-frequency. Adding second-order coefficients only improves results marginally.

Providing adaptivity for the wavelet octave bandwidth $Q_1$ by computing
scattering coefficients for both $Q_1 = 1$ and $Q_1 = 8$ further reduces the
error by almost $10\%$. Indeed, music signals include both sharp transients and
narrow-bandwidth frequency components. We thus have an error rate of $8.6\%$, which compares favorably to the non-scattering state-of-the-art of $9.4\%$ error \cite{lee-shih}.

Replacing the SVM with more sophisticated classifiers can improve results.
A sparse representation classifier applied to
second-order time scattering coefficients reduces 
the error rate from $10.7\%$ to $8.8\%$, as shown in \cite{chen2013}.
Let us mention that the GTZAN database suffers from some significant 
statistical issues \cite{sturm2012analysis}, which probably does not make
it appropriate to evaluate further algorithmic refinements. 

\subsection{Phone Segment Classification}
\label{sec:phone}

The same scattering representation is tested for phone segment classification
with the TIMIT corpus \cite{timit}. The dataset contains $6300$ phrases, each
annotated with the identities, locations, and durations of its constituent
phones. This task is simpler than continuous speech recognition, but 
provides an evaluation of scattering feature vectors for representing phone
segments. Given the location and duration of a phone segment, the goal is to
determine its class according to the standard protocol
\cite{kflee,clarkson-moreno}. The $61$ phone classes (excluding the glottal
stop /q/) are collapsed into $48$ classes, which are used to train and test
models. To calculate the error rate, these classes are then mapped into $39$
clusters. Training is achieved on the full $3696$-phrase training set,
excluding ``SA'' sentences. The Gaussian kernel SVM parameters are optimized
by validation on the standard $400$-phrase development set
\cite{halberstadt-thesis}. The error is then calculated on the core
$192$-phrase test set.

An audio segment of length $192~\mathrm{ms}$ centered on a phone can be
represented as an array of MFCC feature vectors with half-overlapping time
windows of duration $T$. This array, with the logarithm of the phone duration
added, is fed to the SVM. In many cases, hidden Markov models or fixed time
dilations are applied to match different MFCC sequences, to account for the
time-warping of the phone segment \cite{kflee,clarkson-moreno}. Table
\ref{table:results} shows that $T = 23~\mathrm{ms}$ yields a $18.5\%$ error
which is much less than the $60.5\%$ error for $T = 740~\mathrm{ms}$. Indeed,
many phones have a short duration with highly transient structures and are not
well-represented by wide time windows.

A lower error of $17.1\%$ is obtained by replacing the SVM with a sparse
representation classifier on MFCC-like spectral features \cite{sainath}.
Combining MFCCs of different window sizes and using a committee-based
hierarchical discriminative classifier, \cite{chang-glass} achieves an error of $16.7\%$, the best so far. Finally, convolutional deep-belief networks
cascades convolutions, similarly to scattering, on a spectrogram using filters
learned from the training data. These, combined with MFCCs, yield an error of
$19.7\%$ \cite{lee}.

Rows $4$ through $6$ of Table \ref{table:results} gives the classification
results obtained by replacing MFCC vectors with a time scattering transform
computed using first-order wavelets with $Q_1 = 8$. In order to retain local
amplitude structure while creating invariance to loudness changes, first-order
coefficients are renormalized in \eqref{eq:normscatt01} using $|x|$
averaged over a window the size of the whole phone segment. Second- and
third-order scattering coefficients are calculated with $Q_2 = Q_3 = 1$. The
best results are obtained with $T = 32~\mathrm{ms}$. For $\mm = 1$, we only
keep first-order scattering coefficients and get a $19.0\%$ error, similar to
that of MFCCs. The error is reduced by about $10\%$ with $\mm = 2$, a smaller
improvement than for GTZAN because scattering invariants are computed on
smaller time interval $T = 32~\mathrm{ms}$ as opposed to $740~\mathrm{ms}$ for
music. Second-order coefficients carry less energy when $T$ is smaller, as
shown in Table \ref{table:scattering-energy}. For the same reason, third-order
coefficients provide even less information compared to the GTZAN case, and do
not improve results.

Note that no explicit time warping is needed in this model. Thanks to the scattering deformation stability, supervised linear classifiers can indeed compute time-warping invariants which remain sufficiently informative.

For $\mm = 2$, cascading a log-frequency transposition invariance computed
with a first-order frequency scattering transform of Section \ref{sec:cls}
reduces the error by about $5\%$. Computing a second-order frequency
scattering transform only marginally improves results. Allowing to adapt the
wavelet frequency resolution by computing scattering coefficients with $Q_1 =
1$ and $Q_1 = 8$ also reduces the error by about $5\%$

Again, these results are for the problem of phone classification, where boundaries are given. Future work will concentrate on the task of phone recognition, where such information is absent. Since this task is more complex, performance is generally obtained worse, with the state-of-the-art achieved with a $17.7\%$ error rate \cite{graves2013speech}.

\section{Conclusion}

The success of MFCCs for audio classification can partially be explained by
their stability to time-warping deformation. Scattering representations extend
MFCCs by recovering lost high frequencies through successive wavelet
convolutions. Over windows of $T \approx 200~\mathrm{ms}$, signals recovered
from first- and second-order scattering coefficients have a good audio
quality. Normalized scattering coefficients characterizes amplitude
modulations, and are stable to time-warping deformations. A frequency
transposition invariant representation is obtained by cascading a second
scattering transform along frequencies. Time and frequency scattering feature
vectors yield state-of-the-art classification results with a Gaussian kernel
SVM, for musical genre classification on GTZAN, and phone segment
classification on TIMIT.

\appendix

\section{Modulation Spectrum Properties}
\label{append}

Following \eqref{eq:formant-separation10sdf9d}, $a(t)$ is nearly constant over the time support of $\psi_{\lau}$ and
$\widehat h (\om)$ is nearly constant over the frequency support of $\widehat
\psi_{\lau}$. It results that
\begin{equation}
	\label{eq:formant-separation10}
|x \star \psi_\lau (t) |\approx |\widehat h(\lau)| \, 
|e \star \psi_\lau (t)|\, a(t)~.
\end{equation}

Let $e(t)$ be a harmonic excitation. Since we supposed that $\lau/Q_1 \leq \xi$, 
$\widehat \psi_{\lau}$
covers at most one harmonic whose frequency $k \xi$ is close to $\lau$.
It then results from (\ref{eq:formant-separation10}) that
\begin{equation}
	\label{eq:formant-separation}
|x\star\psi_{\lau}(t)|
\approx |\widehat{h}(\lau)|\, |\widehat \psi_\lau (k \xi)|\, a(t)~.
\end{equation}
Computing $S_1x(t,\lau) = |x\star\psi_{\lau}| \star \phi(t)$ gives
\begin{equation}
	\label{eq:first-order-factorization-app}
	S_1x(t,\lau) \approx |\widehat{h}(\lau)|\,
 |\widehat \psi_\lau (k \xi)|\, a\star\phi(t)~.
\end{equation}
Let us compute
\begin{eqnarray*}
	|x|\star\phi(t) &=& \int |e\star h(u)|a(u)\phi(t-u)du\\
	 &=& \frac{2\pi}{\xi} \int 
	\sum_{n=-\infty}^{+\infty} \!\!|h(u+2n\pi/{\xi})| a(u)\phi(t-u)du\\
	 &=& \frac{2\pi}{\xi} \sum_{k=-\infty}^{\infty}\int_0^{2\pi/\xi} 
	\sum_{n=-\infty}^{+\infty} |h(u+{2n\pi}/{\xi})| \\ 
&&a(u+{2k\pi}/{\xi})\phi(t-u-{2k\pi}/{\xi})\,du.
\end{eqnarray*}
Since $a(t)$ and $\phi(t)$ are approximately constant over intervals of size $2 \pi/\xi$, and the support of $h$ is smaller than $2 \pi/\xi$, one can verify that
\begin{equation*}
	|x|\star\phi(t) \approx \|h\|_1\, a \star \phi(t)~.
\end{equation*}
This approximation together with \eqref{eq:first-order-factorization-app} verifies \eqref{eq:first-order-factorization}.

It also results from \eqref{eq:formant-separation} that
\begin{equation*}
S_2 x (t,\lau,\lad) \approx |\widehat{h}(\lau)|\,|\widehat\psi_{\lau}(k \xi)|\,	|a\star\psi_\lad| \star \phi (t) \,,
\end{equation*}
which, combined with \eqref{eq:first-order-factorization-app}, yields \eqref{eq:first-order-factorization3}.

Let us now consider a Gaussian white noise excitation $e(t)$. We saw in (\ref{eq:formant-separation10}) that
\begin{equation}
\label{eq:formant-separation102}
|x \star \psi_\lau (t) |\approx |\widehat h(\lau)| \, 
|e \star \psi_\lau (t)|\, a(t)~.
\end{equation}
Let us decompose
\begin{equation}
\label{asfdns}
|e \star \psi_\lau (t)| = \mathbb{E}(|e\star\psi_{\lau}|) + \epsilon(t)~,
\end{equation}
where $\epsilon (t)$ is a zero-mean stationary process. If $e(t)$ is a
normalized Gaussian white noise then $e \star \psi_{\lau}(t)$ is a Gaussian random
variable of variance $\|\psi_{\lau}\|^2$. It results
that $|e \star \psi_{\lau}(t)|$ and $\epsilon (t)$ have a Rayleigh
distribution, and since $\psi$ is a complex wavelet with quadrature phase, 
one can verify that
\[
\mathbb{E}(|e \star \psi_{\lau}|)^2 = \frac \pi  4\, 
\mathbb{E}(|e \star \psi_{\lau}|^2) = \frac \pi  4  \,\|\psi_\lau\|^2 ~.
\]
Inserting (\ref{asfdns}) and this equation in (\ref{eq:formant-separation102})
shows that
\begin{equation}
	\label{eq:first-order-separation}
|x \star \psi_{\lau}(t)| \approx |\widehat h(\lau)| 
\Bigl(\pi^{1/2} 2^{-1} \|\psi_\lau\| a(t) + a(t)\,\epsilon (t) \Bigr)~.
\end{equation}
When averaging with $\phi$, we get
\begin{equation}
S_1 x(t,\lau) \approx |\widehat h(\lau)|\, 
\Bigl(\pi^{1/2} 2^{-1} \|\psi_\lau\| a \star \phi(t) + (a\,\epsilon) \star
\phi (t) \Bigr)~.
\end{equation}
Suppose that $a(t)$ is not
sparse, in the sense that 
\begin{equation}
	\label{eq:non-sparsity}
	\frac{|a|^2\star\phi(t)}{|a\star\phi|^2(t)} \sim 1~.
\end{equation}
It means that ratios between local $\mathbf{L}^2$ and $\mathbf{L}^1$ norms of
$a$ is of the order of $1$. We are going to show that if $T^{-1} \ll \lau
Q^{-1}_1$ then
\begin{equation}
	\label{eq:aephi-error}
	\frac{\mathbb{E}(|(a\,\epsilon)\star\phi(t)|^2)}{\|\psi_\lau\|^2\, |a \star \phi(t)|^2} \ll 1
\end{equation}
which implies
\begin{equation}
	\label{eq:noise-first-order-approx}
	 S_1 x(t,\lau) \approx \frac{\pi^{1/2}}{2} \|\psi\|\,\lau^{1/2}\,|\widehat h (\lau)|\, a\star\phi(t)~.
\end{equation}
We give the main arguments to compute the order of magnitudes of the
stochastic terms, but it is not a rigorous proof. For a detailed argument, see \cite{anden-thesis}.
Computations rely on the following lemma.
\begin{lemma}
\label{lamsdf}
Let $z(t)$ be a zero-mean stationary process of power spectrum
$\widehat R_z (\omega)$. For any deterministic functions $a(t)$ and $h(t)$
	\begin{equation}
		\mathbb{E}(|(z a) \star h(t)|^2) \leq \sup_\omega \widehat R_z(\omega)\, |a|^2 \star |h|^2 (t)~.
	\end{equation}
\end{lemma}

\begin{proof}
Let $R_z(\tau) = \E(z(t)\, z(t+\tau))$,
\[
\mathbb{E}(|(z a) \star h(t)|^2) =
\iint R_z (v-u) \, a(u) \,h(t-u)\, a(v)^*\,h(t-v)^*\,  du dv 
\]
and hence
\[
\mathbb{E}(|(z a) \star h(t)|^2) = 
\langle R_z y_t , y_t \rangle~\mbox{with}~y_t (u) = a(u) h(t-u). 
\]
Since $R_z$ is the kernel of a positive symmetric operator whose
spectrum is bounded by $\sup_\omega \widehat R_z(\omega)$ it results that
\[
\mathbb{E}(|(z a) \star h(t)|^2) \leq 
\sup_\omega \widehat R_z(\omega)\, \|y_t\|^2 = 
\sup_\omega \widehat R_z(\omega)\, |a|^2 \star |h|^2 (t)~.
\]
\end{proof}

Because $e(t)$ is a normalized white noise, with a Gaussian
chaos expansion, one can verify \cite{anden-thesis} that $\sup_\omega \widehat
R_{\epsilon} (\omega) \leq C(1 - \pi / 4)$, where $C = \|\psi\|_1^2 \approx 1$. Applying Lemma \ref{lamsdf} to $z = \epsilon$ and $h = \phi$ gives
\[
\E(|(\epsilon\,a)\star\phi(t)|^2) \leq (1 - \pi/4)
\,|a|^2 \star |\phi|^2 (t)~.
\]
Since $\phi$ has a duration $T$, it can be written as $\phi(t) = T^{-1}\phi_0 (T^{-1}t)$ for some $\phi_0$ of duration $1$. As a result, if \eqref{eq:non-sparsity} holds then
\begin{equation}
	\label{eq:variance-ratio}
	\frac{|a|^2 \star |\phi|^2 (t)} 
	{|a \star \phi (t)|^2} \sim \frac{1}{T}
\end{equation}
The frequency support of $\psi_\lau$ is proportional to $\lau Q_1^{-1}$, so we have $\|\psi_\lau\|^2 \sim \lau Q_1^{-1}$. Together with \eqref{eq:variance-ratio}, if $T^{-1} \ll \lau  Q^{-1}_1$ it proves \eqref{eq:aephi-error} which yields \eqref{eq:noise-first-order-approx}.

We approximate $|x|\star\phi(t)$ similarly. First, we write
\begin{equation}
	|e\star h(t)| = \E|e\star h|+\epsilon'(t),
\end{equation}
where $\epsilon'(t)$ is a zero-mean stationary process. Since $e\star h(t)$ is normally distributed in $\mathbb{R}$, $|e\star h|(t)$ has $\chi^1$ distribution and
\begin{equation}
	\E(|e\star h|)^2 = \frac{2}{\pi}\E(|e\star h|^2) = \frac{2}{\pi}\|h\|^2,
\end{equation}
which then gives
\begin{equation}
	|x|\star\phi(t) = \sqrt{\frac{2}{\pi}}\|h\|a\star\phi(t) + (a\epsilon')\star\phi(t).
\end{equation}

One can show that $\mathrm{sup}_\om \widehat{R}_{\epsilon'}(\om) \le (1-2/\pi)\|h\|_1^2$ \cite{anden-thesis}, so applying Lemma \ref{lamsdf} gives
\begin{equation}
	\E \left( \left| (a\epsilon')\star\phi(t) \right|^2 \right) \le (1-2/\pi)\|h\|_1^2 |a|^2\star|\phi|^2(t).
\end{equation}
Now \eqref{eq:variance-ratio} implies that
\begin{equation}
	\frac{\E \left( \left| (a\epsilon')\star\phi(t) \right|^2 \right)} {\|h\|^2|a\star\phi(t)|^2} \ll 1
\end{equation}
since $a$ is non-sparse and because $h$ has a support much smaller than $T$ so 
$\|h\|_1^2/\|h\|^2 \ll T$. Consequently,
\begin{equation}
	|x|\star\phi(t) \approx \sqrt{\frac{2}{\pi}}\|h\|a\star\phi(t),
\end{equation}
which, together with \eqref{eq:noise-first-order-approx} gives \eqref{eq:first-order-factorization0}.

Let us now compute $S_2 x(t,\lau,\lad) = ||x \star \psi_\lau| \star \psi_\lad|
\star \phi(t)$. If $T^{-1} \ll \lau Q^{-1}_1$ then
\eqref{eq:noise-first-order-approx} together with
\eqref{eq:first-order-separation} shows that
\begin{equation}
	\label{eq:second-order-factorization0}
\frac{S_2 x(t,\lau,\lad)} {S_1 x(t,\lau)} \approx 
\frac{|a\star\psi_\lad| \star \phi (t)}
{a\star \phi (t)} + \widetilde \epsilon(t) \,,
\end{equation}
where
\begin{equation}
	\label{eq:second-order-factorization0sdf}
0 \leq \widetilde\epsilon(t) \leq \frac{2|(a\epsilon)\star\psi_\lad|\star \phi (t)}{\pi^{1/2} \|\psi_\lau\| \,a \star \phi(t)} ~.
\end{equation}
Observe that
\begin{equation*}
E(|(a\epsilon)\star\psi_\lad|\star \phi (t)) \le
E(|(a\epsilon)\star\psi_\lad|^2)^{1/2}\star\phi(t) .
\end{equation*}
Lemma \ref{lamsdf} applied to $z = \epsilon$ and
$h = \psi_\lad$ gives the following upper bound:
\begin{equation}
	\label{eq:second-order-factorization0sdf9}
\E(|(a \epsilon) \star \psi_{\lad}  (t)|^2) \leq C(1-\pi/4)\, 
|a|^2 \star |\psi_{\lad}|^2(t)~.
\end{equation}
One can write $|\psi_{\lad}(t)|= \la_2 Q_2^{-1} \theta( \la_2 Q_2^{-1} t)$
where $\theta(t)$ satisfies $\int \theta(t)\,dt \sim 1$. 
Similarly to \eqref{eq:variance-ratio}, if \eqref{eq:non-sparsity} holds over time intervals of size $Q_2/\lad$, then
\begin{equation}
	\label{eq:variance-ratio2}
\frac{|a|^2 \star |\psi_{\lad}|^2(t)}{|a \star |\psi_\lad||^2} \sim \frac{\lad}{Q_2}.
\end{equation}
Since $\|\psi_{\lau} \|^2 \sim \lau Q_1^{-1}$ and $|\psi_\lad|\star\phi(t)
\sim \phi(t)$ when $Q_2/\lad \le T$, it results from
(\ref{eq:second-order-factorization0sdf},\ref{eq:second-order-factorization0sdf9},\ref{eq:variance-ratio2})
that $0 \leq \mathbb{E}(\widetilde\epsilon(t)) \leq C\,(4/ \pi - 1)^{1/2}
(\lad\, Q_1)^{1/2}(\lau\, Q_2)^{-1/2}$ with $C \sim 1$.



\ifCLASSOPTIONcaptionsoff
  \newpage
\fi

\bibliographystyle{IEEEtran}
\bibliography{IEEEabrv,paper}
%



%


\end{document}